\documentclass[journal]{IEEEtran}

\usepackage{cite}
\usepackage{algorithmic}
\usepackage{graphicx}
\usepackage{textcomp}
\usepackage{xcolor}
\usepackage{amsmath,amsthm,amsfonts,amssymb,amscd}
\usepackage{lastpage}
\usepackage{enumerate}
\usepackage{fancyhdr}
\usepackage{mathrsfs}
\usepackage{listings}
\usepackage{hyperref}
\usepackage{comment}
\usepackage[english]{babel}

\usepackage{stfloats}
\usepackage{ragged2e}

\newtheorem{theorem}{Theorem}

\newtheorem{corollary}{Corollary}
\newtheorem{remark}{Remark}
\newtheorem{example}{Example}
\newtheorem{proposition}{Proposition}

\def\BibTeX{{\rm B\kern-.05em{\sc i\kern-.025em b}\kern-.08em
    T\kern-.1667em\lower.7ex\hbox{E}\kern-.125emX}}

\begin{document}

\title{Capacity Bounds for Broadcast Channels with Bidirectional Conferencing Decoders}

\author{Reza K. Farsani and Wei Yu, \IEEEmembership{Fellow, IEEE}%
\thanks{
Manuscript submitted to {\it IEEE Transactions on Information Theory} on June 28, 2024, revised on February 2, 2025, and April 6, 2025. 
%The date of this version is \today. 
This work is supported by the Natural Sciences and Engineering Research Council (NSERC) of Canada via a Discovery Grant.
The materials in this paper have been presented in part at the IEEE Information Theory Workshop (ITW), Saint Malo, France, April 2023 \cite{reza_itw} and in part at the IEEE International Symposium on Information Theory (ISIT), Taiwan, June 2023 \cite{reza_isit}. 
	The authors are with The Edward S.\ Rogers Sr.\ 
Department of Electrical and Computer Engineering, University of Toronto, 
10 King's College Road, Toronto, Ontario M5S3G4, Canada.
E-mails: \{rkfarsani, weiyu\}@ece.utoronto.ca. 
}
}
\maketitle

\thispagestyle{empty}

\allowdisplaybreaks

\begin{abstract}
The two-user broadcast channel (BC) with receivers connected by bidirectional cooperation links of finite capacities, known as conferencing decoders, is considered. A novel capacity region outer bound is established based on multiple applications of the Csisz\'{a}r-K\"{o}rner identity. Achievable rate regions are derived by using Marton's coding as the transmission scheme, together with different combinations of decode-and-forward and quantize-bin-and-forward strategies at the receivers. It is shown that the outer bound coincides with the achievable rate region for a new class of semi-deterministic BCs with degraded message sets; for this class of channels, one-round cooperation is sufficient to achieve the capacity. Capacity result is also derived for a class of more capable semi-deterministic BCs with both common and private messages and one-sided conferencing. For the Gaussian BC with conferencing decoders, if the noises at the decoders are perfectly correlated (i.e., the correlation is either 1 or -1), the new outer bound yields exact capacity region for two cases: i) BC with degraded message sets; ii) BC with one-sided conferencing from the weaker receiver to the stronger receiver.  When the noises have arbitrary correlation, the outer bound is shown to be within half a bit from the capacity region for these same two cases. Finally, for the general Gaussian BC, a one-sided cooperation scheme based on decode-and-forward from the stronger receiver to the weaker receiver is shown to achieve the capacity region to within $\frac{1}{2}\log (\frac{2}{1-|\lambda|})$ bits, where $\lambda$ is the noise correlation. An interesting implication of these results is that for a Gaussian BC with perfectly negatively correlated noises and conferencing decoders with finite cooperation link capacities, it is possible to achieve a strictly positive rate using only an infinitesimal amount of transmit power.
\end{abstract}


\section{Introduction}

In many communication networks, it is sometimes feasible for the receivers 
at distinct locations to exchange messages and to cooperate. Receiver
cooperation is known to be able to improve the performance of the direct
communication from the transmitter.  This paper considers a channel model
in which cooperation takes place via dedicated digital links of finite
capacities---termed \textit{conferencing links}.  
We investigate the impact of receiver cooperation on the capacity region
of a two-user broadcast channel (BC) with conferencing links, with both
common and private messages. 

\subsection{Motivation} 

We are particularly interested in the effect of receiver noise correlation 
on the capacity of BC with conferencing decoders. To motivate the discussion, consider 
a Gaussian channel model as shown in Fig.~\ref{fig:Gaussian_BC}. 
As is well known, the capacity of a single-user Gaussian channel is 
$C=\frac{1}{2}\log\left( 1 + \frac{P}{N}\right)$ bits per channel use, 
where $P$ is the average transmit power constraint and $N$ is 
the receiver noise power.  
Thus, at a fixed positive noise power $N$, if the transmit
power $P \to 0$, then $C \to 0$.

Now consider the scenario of transmitting a common message through
a two-user Gaussian broadcast channel with conferencing decoders. 
If the noises at the two receivers have the same power and they are 
\emph{uncorrelated}, then full receiver cooperation can potentially 
increase the signal-to-noise ratio (SNR) by 3dB, resulting in up to 
1 bit increase in the common message capacity of the BC. 

The situation is completely different if the receiver noises are
\emph{correlated}.  Specifically, if the receiver noises are 
\emph{perfectly positively} correlated, then cooperation would gain nothing, 
because $Y_1 = Y_2$ in this case; as each receiver already knows 
what the other receiver knows, exchanging information does not help. 
However, if the receiver noises are \emph{perfectly negatively} correlated, 
then cooperation can help significantly. Consider the case where 
$Y_1 = X + Z$ and $Y_2=X-Z$, and the conferencing links have infinite capacities
$C_{12} = C_{21} = \infty$. 
Such a BC with \emph{full receiver cooperation} has
\emph{infinite} capacity at any positive transmit power $P$, because having
access to both $Y_1$ and $Y_2$ allows the noises to be cancelled completely.

What happens if the noises are perfectly negatively correlated, but the conferencing 
links have \emph{finite} capacity, e.g., $C_{12} = C_{21} = 1$ bit per channel use? 
Complete noise cancellation is no longer possible. 
The analysis becomes more delicate, and this is one of the main topics of this
paper. It turns out that
an interesting phenomenon emerges---it is possible to achieve a strictly positive
rate (i.e., 1 bit in the example above), even with an infinitesimally small amount 
of transmit power.  In contrast to the single-user Gaussian
channel where $P \to 0$ implies $C \to 0$, for this Gaussian BC with perfectly 
negatively correlated noises and conferencing decoders of fixed cooperative
link capacities, we can have $P \to 0$, while $C \to 1$. 
In other words, the energy efficiency for transmitting one bit of information 
can be infinite!

\begin{figure*}
    \centering
    \includegraphics[width=0.7\textwidth]{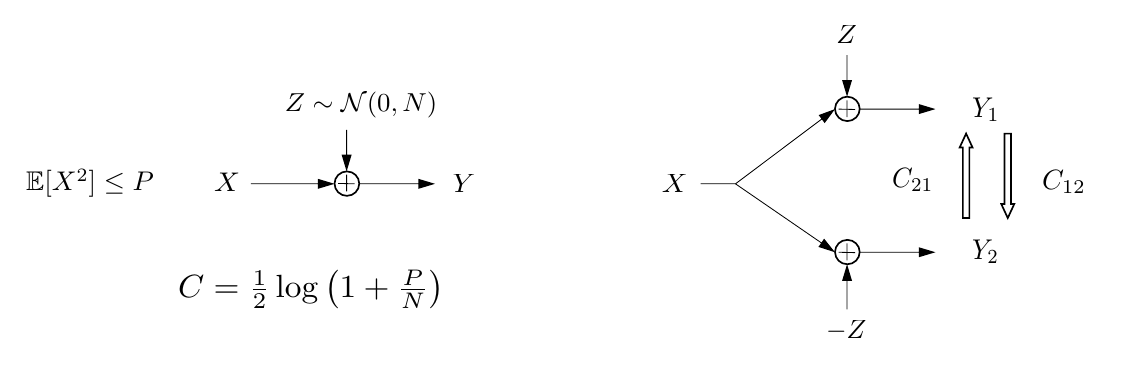}
	\caption{The capacity of a Gaussian channel (left) goes 
	to zero if the transmit power $P$ goes to zero at a fixed noise
variance $N$.  In a Gaussian BC with perfectly negatively correlated noises at the
receiver (right), conferencing links of fixed capacities $C_{12}$ and 
$C_{21}$ between the receivers can allow a strictly positive rate to be
achieved even at an infinitesimally small amount of transmit power.}
    \label{fig:Gaussian_BC}
\end{figure*}

The above example motivates us to study the BC with conferencing decoders
with a specific focus on the Gaussian channel with correlated noises. 
We remark that correlated noises often occur in practical communication scenarios
when the noises at the two receivers are affected by a strong common interference source. %A common interference source affecting both receivers. % would result in correlation in the noises.

\subsection{Prior Literature} 

Cooperation via conferencing links has been studied for many different
communication networks in the information theory literature \cite{Wimac, Madsen, 
CaoChen, Maric, Gunduz, SkoglundA, SkoglundI, Rezamac1, Tse1, Tse2, Kang, Rezamac2,
Wei, SkoglundT, RezaK1, RezaK2}.  Further, the papers \cite{Draper, Dabora, bross, Dik, 
Gold, Cuff, absent, unreliable} specifically consider the two-user BC with 
conferencing decoders. In \cite{Draper}, the authors develop communication
strategies for the interactive decoding of a common message broadcast to
cooperative receivers. In \cite{Dabora}, the capacity region of physically degraded
channel is derived and also an achievable rate region is given for the general
case. In \cite{bross}, an achievable rate region is presented based on coding strategies for the partially cooperative relay broadcast channels and also a converse result is proved. In \cite{Dik}, the problems of communication over physically degraded,
state-dependent BCs with one-sided conferencing decoders are investigated. In
\cite{Gold}, the capacity region of the semi-deterministic BC with one-sided
decoder cooperation is derived and its duality with a source coding problem is
addressed. The authors in \cite{Cuff} consider the BC with one-sided
cooperating receivers under the strong secrecy constraints and present capacity 
results for semi-deterministic and physically degraded cases. In \cite{absent},
the BC with (one-sided) unreliable cooperating decoders is studied. 
In a recent work \cite{unreliable}, the BC with degraded message sets and
one-sided cooperation link that may be absent is considered and its capacity
region is given.

\subsection{Main Contributions} 

The existing capacity results for the BC with conferencing decoders are all for the case of one-sided cooperation, i.e., only one of the decoders is connected to the other by a cooperating link. This is due to the lack of useful outer bounds (beyond the cut-set bound) for the two-sided cooperation case. 
In this paper, we first establish a novel outer bound on the capacity region of the two-user BC with bidirectional conferencing decoders. The new outer bound, which is derived using multiple applications of the Csisz\'{a}r-K\"{o}rner identity \cite{marton} \cite[Lemma 7]{CsisKo}, is strictly tighter than the previous ones including that of \cite[Proposition 1]{Dabora}, which is essentially the cut-set bound. 

We then propose achievability strategies for the BC with conferencing decoders. In \cite [Theorem 2]{Dabora} an achievable region is derived for the BC with bidirectional cooperation by applying Marton's coding at the transmitter and the compress-and-forward cooperative scheme at both receivers. Another achievable region is given in \cite{bross} based on coding
strategies for the partially cooperative relay broadcast channels. A third achievable region is given in \cite[Appendix B]{Gold} for the BC with one-sided cooperation between the receivers, which is derived by applying Marton's coding at the transmitter and decode-and-forward as the cooperative protocol. All of these achievable schemes are however insufficient to either derive new capacity results for the discrete memoryless BCs, or approximate capacity results for the Gaussian BC. 

This paper presents two achievability schemes for the two-user BC with both common and private messages and bidirectional conferencing decoders. In the first scheme, we apply Marton's coding as the transmission scheme, and quantize-bin-and-forward at one receiver first and then a combination of decode-and-forward and quantize-bin-and-forward at the other receiver as cooperative strategy. In the second scheme, we apply a combination of decode-and-forward and quantize-bin-and-forward at one receiver first and then quantize-bin-and-forward at the other receiver as cooperation scheme. For each achievability scheme, at the receiver that applies a combination of decode-and-forward and quantize-bin-and-forward, the optimal proportion of the capacity of the cooperative link that should be devoted to each strategy is determined. 

We prove that the novel outer bound coincides with the first achievable rate region for a new class of semi-deterministic BCs with degraded message sets. This capacity result is important from two viewpoints. First, it is the first capacity result for the two-user BC with two-sided conferencing decoders (all previously known capacity results are regarding the channel with one-sided cooperation). Second, it is among the rare cases in network information theory for which quantize-bin-and-forward is optimal. These results also demonstrate that a one-round cooperation protocol is sufficient to achieve capacity and multi-round strategies similar to those devised in \cite{Draper} are not needed. Moreover, we derive a capacity result for a new class of more capable semi-deterministic BCs with both common and private messages and one-sided cooperation. %We provide examples to explain the results.

Furthermore, we evaluate the derived outer bound for the Gaussian channel with correlated noises. Using this bound, we prove several interesting results. For the channel with perfectly correlated noises (when noise correlation is either 1 or -1), the new outer bound yields the exact capacity region for two cases: i) BCs with degraded message sets; ii) BCs with one-sided conferencing between decoders. For these two cases, we also show that the outer bound is within half a bit from the capacity region over arbitrary noises correlation. 
Lastly, we prove that for the Gaussian BC, a one-sided cooperative scheme from the stronger receiver to the weaker receiver based on the decode-and-forward technique is already sufficient to achieve the capacity region to within $\frac{1}{2}\log (\frac{2}{1-\left|\lambda\right|})$ bits, where $\lambda$ is the correlation coefficient of the noises at the receivers of the BC.
Therefore, such a strategy is approximately optimal when the noise correlation is small.

\subsection{Organization of Paper and Notations}

The rest of the paper is organized as follows. We begin by defining the channel model in Section \ref{sec:channel_model}. The main converse and achievability results for the discrete memoryless BC are presented in Section \ref{sec:dmc}. Section \ref{sec:gaussian} treats the Gaussian BC with conferencing decoders. Section \ref{sec:conclude} concludes the paper.  The appendices contain most of the proofs.

In this paper, we use the following notations. A random variable is given by upper case letter (e.g. $X$) and its realization is shown by lower case letter (e.g. $x$). 
We use $X^n$ to denote a sequence of random variables $(X_1,\cdots,X_n)$ and use the notation $X_t^n=(X_t, X_{t+1},\cdots,X_n)$.
Further, for integers $m$ and $n$, we use $[m:n]$ to denote the index set $\{m, \cdots, n\}$.

The probability density function (PDF) or the probability mass function (PMF) 
of a random variable $X$ is denoted by $P_X(x)$ and the conditional PDF/PMF of $X$ given $Y$ is denoted by $P_{X|Y}(x|y)$, where the subscripts are omitted occasionally for brevity. The operator $\{a\}^+$ is defined as: $\{a\}^+=\max\{0,a\}$. The set of non-negative real numbers is given by $\mathbb{R}_+$. For $\epsilon>0$, the set of all jointly $\epsilon$-letter typical $n$-sequences $x^n$ with respect to the PDF/PMF $P_X(x)$ is denoted by $\mathcal{T}^n_{\epsilon}(P_X)$; (see \cite{NIT} for definition).
We use the notation $X=\varnothing$ to denote that
a random variable $X$ is set as a constant.
Finally, we use the shorthand $\psi(x) = \frac{1}{2} \log(1+x)$.

\begin{figure*}
    \centering
    \includegraphics[width=0.67\textwidth]{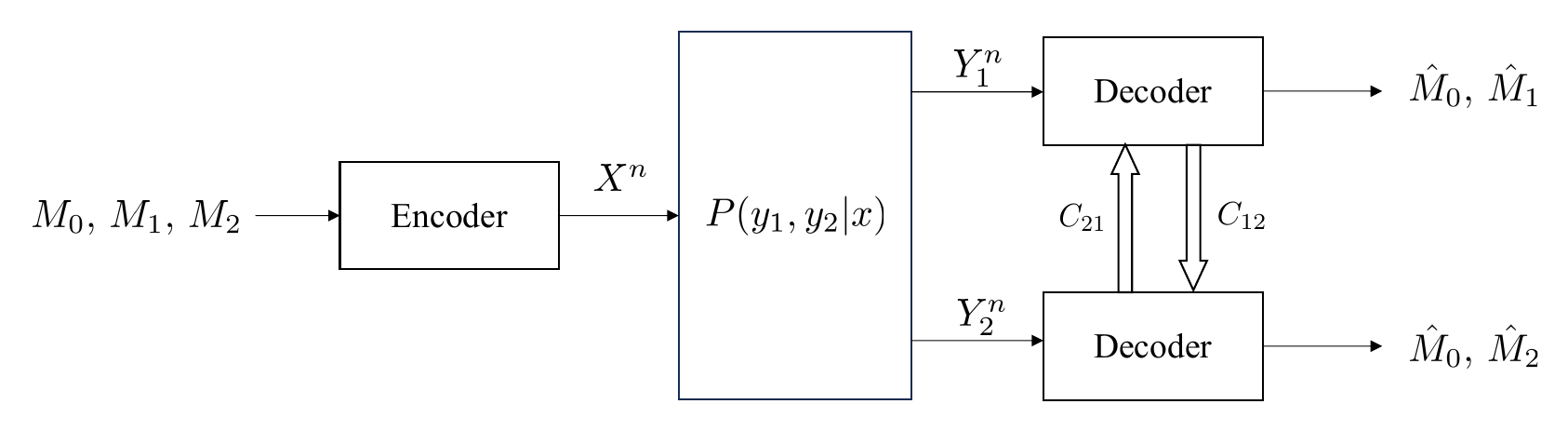}
    \caption{Two-user BC with conferencing decoders}
    \label{fig:BC_CD}
\end{figure*}

\section{Channel Model}
\label{sec:channel_model}

%\begin{definition}
The two-user BC with conferencing decoders is a communication scenario in which a transmitter sends a common message and two private messages to two receivers, and the two receivers are able to exchange information via communication links of finite capacities (called conferencing links). Fig.~\ref{fig:BC_CD} illustrates the channel model. The channel is given by $(\mathcal{X},\mathcal{Y}_1,\mathcal{Y}_2, P(y_1,y_2|x),C_{12},C_{21})$ where $\mathcal{X}$ denotes the input alphabet, $\mathcal{Y}_1$ and $\mathcal{Y}_2$ denote the output alphabets, $P(y_1,y_2|x)$ is the channel transition probability, and $C_{12}$ and $C_{21}$ are capacities of the conferencing links per channel use. 
%\end{definition}

%\textit{Enoding:} 
For the BC with conferencing decoders, a length-$n$ code with $\Gamma$ conferencing rounds is described as follows.
The transmitter encodes independent messages $M_0$, $M_1$,  and $M_2$, which are uniformly distributed over the sets $[1:2^{nR_0}]$, $[1:2^{nR_1}]$, and $[1:2^{nR_2}]$, respectively, into a codeword and sends over the channel according to the following:
\begin{align*}
    &\Delta: [1:2^{nR_0}]\times[1:2^{nR_1}]\times[1:2^{nR_2}]\longmapsto \mathcal{X}^n  \\
	&X^n=\Delta(M_0,M_1,M_2).
\end{align*}
The receiver $Y_i,\;i=1,2$, receives a sequence $Y_i^n\in\mathcal{Y}_i^n$. The code consists of two sets of conferencing functions $\{\Xi_{12,\gamma}\}_{\gamma=1}^\Gamma$ and $\{\Xi_{21,\gamma}\}_{\gamma=1}^\Gamma$ with the corresponding output alphabets $\{\mathcal{J}_{12,\gamma}\}_{\gamma=1}^\Gamma$ and $\{\mathcal{J}_{21,\gamma}\}_{\gamma=1}^\Gamma$, respectively, which are described as follows:
\begin{align*}
    \Xi_{12,\gamma}&:\mathcal{Y}_1^n\times\mathcal{J}_{21,1}\times...\times\mathcal{J}_{21,\gamma-1}\longmapsto\mathcal{J}_{12,\gamma} \\
    &J_{12,\gamma}=\Xi_{12,\gamma}(Y_1^n,J_{21}^{\gamma-1}) \\
    \Xi_{21,\gamma}&:\mathcal{Y}_2^n\times\mathcal{J}_{12,1}\times...\times\mathcal{J}_{12,\gamma-1}\longmapsto\mathcal{J}_{21,\gamma} \\
    &J_{21,\gamma}=\Xi_{21,\gamma}(Y_2^n,J_{12}^{\gamma-1}).
\end{align*}
A sequence of conferencing rounds is said to be $(C_{12},C_{21})$-permissible if
\begin{align*}
    \sum_{\gamma=1}^\Gamma \log |\mathcal{J}_{12,\gamma} |\le nC_{12},\qquad \sum_{\gamma=1}^\Gamma \log |\mathcal{J}_{21,\gamma} |\le nC_{21}.
\end{align*}
Before decoding, the receivers exchange information by holding a $(C_{12},C_{21})$-permissible conference. The first receiver obtains the sequence $J_{21}^\Gamma=(J_{21,1},J_{21,2},...,J_{21,\Gamma})$ and the second one obtains the sequence $J_{12}^\Gamma=(J_{12,1},J_{12,2},...,J_{12,\Gamma})$. The receivers then decode their respective messages based on the following decoding functions:
\begin{align*}
    \nabla_1&:\mathcal{Y}_1^n\times \mathcal{J}_{21}^\Gamma\longmapsto[1:2^{nR_0}]\times[1:2^{nR_1}] \\
&(\hat{M}_0,\hat{M}_1)=\nabla_1(Y_1^n\times J_{21}^\Gamma)\\
    \nabla_2&:\mathcal{Y}_2^n\times \mathcal{J}_{12}^\Gamma\longmapsto[1:2^{nR_0}]\times[1:2^{nR_2}] \\
&(\hat{M}_0,\hat{M}_2)=\nabla_2(Y_2^n\times J_{12}^\Gamma).
\end{align*}
The capacity region for the channel is defined as usual \cite{NIT}. Here, we omit the details for brevity.

\section{Discrete Memoryless BC with Conferencing Decoders}
\label{sec:dmc}

\subsection{Converse}

\begin{theorem}\label{th.obc}
Consider the two-user BC with conferencing decoders shown in Fig.~\ref{fig:BC_CD}. Let $\mathcal{R}_o$ denote the set of all rate triples $(R_0,R_1,R_2)$ such that
\begin{align} %\label{eq:obc}
    R_0+R_1 &\le I(U;Y_1)+C_{21} \label{eq:outer_R1U}\\
    R_1 &\le I(X;Y_1|Y_2,V)+I(X;Y_2) \label{eq:outer_R1V_1}\\
    R_1 &\le I(X;Y_2|Y_1,V)+I(X;Y_1) \label{eq:outer_R1V_2}\\
    R_0+R_2 &\le I(V;Y_2)+C_{12} \label{eq:outer_R2V}\\
    R_2 &\le I(X;Y_2|Y_1,U)+I(X;Y_1) \label{eq:outer_R2U_1}\\
    R_2 &\le I(X;Y_1|Y_2,U)+I(X;Y_2) \label{eq:outer_R2U_2}\\
    R_0+R_1+R_2 &\le I(X;Y_1|V)+I(V;Y_2)+C_{12}+C_{21} \label{eq:outer_RsumV_1}\\
    R_0+R_1+R_2 &\le I(X;Y_2|U)+I(U;Y_1)+C_{12}+C_{21} \label{eq:outer_RsumU_1}\\
    R_0+R_1+R_2 &\le I(X;Y_1|Y_2,V)+I(X;Y_2)+C_{12} \label{eq:outer_RsumV_2}\\
    R_0+R_1+R_2 &\le I(X;Y_2|Y_1,U)+I(X;Y_1)+C_{21} \label{eq:outer_RsumU_2}\\
    R_0+R_1+R_2 &\le I(X;Y_1,Y_2) \label{eq:outer_cs}
\end{align}
for some joint PDFs $P(u,v,x)$, where $U,V \rightarrow X \rightarrow Y_1,Y_2$ forms a Markov chain and $|\mathcal{U}|\le |\mathcal{X}|+2$ and $|\mathcal{V}|\le |\mathcal{X}|+2$.
Then $\mathcal{R}_o$ constitutes an outer bound on the capacity region.
\end{theorem}

\begin{IEEEproof}
%{Proof of Theorem \ref{th.obc}:} 
	The proof is given in Appendix \ref{app:outer_bound}. %$\blacksquare$
\end{IEEEproof}

This novel outer bound on the capacity region of the channel is derived through
multiple applications of the Csisz\'{a}r-K\"{o}rner identity \cite{marton}
\cite[Lemma 7]{CsisKo}. 
The Csisz\'{a}r-K\"{o}rner identity is used to establish the constraints 
\eqref{eq:outer_RsumV_1}, \eqref{eq:outer_RsumU_1}, \eqref{eq:outer_RsumV_2},
and \eqref{eq:outer_RsumU_2}
%(\ref{eq:outer_R1U}), (\ref{eq:outer_R2V}), (\ref{eq:outer_RsumV_1}) and (\ref{eq:outer_RsumU_1}) 
based on the appropriate definitions of auxiliary random variables $U$ and $V$. 

Specifically, the definitions of $U$ and $V$ are carefully chosen to ensure a compact representation of the outer bound using only two auxiliary random variables. 
Moreover, the multi-letter mutual information expressions are carefully manipulated 
in order to establish novel structures in the constraints
(\ref{eq:outer_R1V_1}), (\ref{eq:outer_R1V_2}), (\ref{eq:outer_R2U_1}), 
(\ref{eq:outer_R2U_2}), (\ref{eq:outer_RsumV_2}) and (\ref{eq:outer_RsumU_2}) in terms of these auxiliary variables. 
%which  consistently include the same auxiliaries are completely novel in this paper. 
Note that the constraint (\ref{eq:outer_cs}) is due to the cut-set bound. 

As we demonstrate later, the novel structure of the outer bound is crucial 
for deriving new capacity results for specific channels (and also approximate capacity results for specific Gaussian channels). 

The outer bound $\mathcal{R}_o$ is tighter than that of \cite[Proposition 1]{Dabora}, which is essentially the cut-set bound.

\subsection{Achievability}

%As mentioned earlier, in \cite[Theorem 2]{Dabora} and \cite[Appendix B]{Gold}, achievable rate regions are given for the BC with conferencing receivers.  However, those regions (which are given for the channel with private messages only) are in general inefficient to either derive new capacity results or approximate capacity results for the Gaussian channel. Here, 
We now present novel achievability
schemes for the two-user BC with both common and private messages and
bidirectional conferencing decoders. The achievability schemes consist
of judicious combination of Marton's coding together with decode-and-forward
and quantize-bin-and-forward cooperation strategies.
%In the first scheme, we apply Marton’s coding as the transmission scheme, and quantize-bin-and-forward at one receiver first and then a combination of decode-and-forward and quantize-bin-and-forward at the other receiver as cooperative strategy. 
%In the second scheme, we apply a combination of decode-and-forward and quantize-bin-and-forward at one receiver first and then quantize-bin-and-forward at the other receiver as cooperation strategy. 
The achievable rate regions are given in the following results.

\begin{proposition}\label{th.ibc1}

Consider the two-user BC with conferencing decoders shown in Fig.~\ref{fig:BC_CD}. Let $\mathcal{R}_i^{(1)}$ denote the set of all rate triples $(R_0,R_1,R_2)\in \mathbb {R}_+^3$ such that
\begin{align}
\label{eq:ibc1}
R_0+ R_1 &\le \min \Big\{I(U,W;Y_1)+\zeta_2, I(U,W;Y_1,\hat{Y}_2)\Big\}\\
R_0+ R_2 &\le \min \Big\{I(V,W;Y_2)+\zeta_1+\Bar{\alpha}_1C_{12}, \nonumber \\
 & \qquad I(V,W;\hat{Y}_1,Y_2)+\Bar{\alpha}_1C_{12}\Big\}\\
R_0+R_1+ R_2 &\le \min \Big\{I(U;Y_1|W)+\zeta_2, I(U;Y_1,\hat{Y}_2|W)\Big\} \nonumber \\
 & \quad +\min \Big\{I(V,W;Y_2)+\zeta_1+\Bar{\alpha}_1C_{12}, \nonumber \\
 & \qquad \qquad \qquad I(V,W;\hat{Y}_1,Y_2)+\Bar{\alpha}_1C_{12}\Big\} \nonumber \\
& \quad -I(U;V|W) \\
R_0+R_1+ R_2 &\le  \min \Big\{I(U,W;Y_1)+\zeta_2, I(U,W;Y_1,\hat{Y}_2)\Big\}\nonumber\\
 & \quad + \min \Big\{I(V;Y_2|W)+\zeta_1, \nonumber \\ 
 & \qquad I(V;\hat{Y}_1,Y_2|W)\Big\} -I(U;V|W) \\
2R_0+R_1+ R_2 &\le  \min \left \{I(U,W;Y_1)+\zeta_2, I(U,W;Y_1,\hat{Y}_2)\right \}\nonumber\\
 & \quad + \min \Big\{I(V,W;Y_2)+\zeta_1+\Bar{\alpha}_1C_{12}, \nonumber \\ 
 & \qquad \qquad \qquad I(V,W;\hat{Y}_1,Y_2)+\Bar{\alpha}_1C_{12}\Big\} \nonumber \\ &\quad -I(U;V|W)
\label{eq:ibc1_end}
\end{align}
where
\begin{align}
    \zeta_1&=\{\alpha_1 C_{12}-I(\hat{Y}_1;U,Y_1|V,W,Y_2)\}^+\\
    \zeta_2&=\{C_{21}-I(\hat{Y}_2;Y_2|U,W,Y_1)\}^+
\end{align}
for some joint PDFs $P(u,v,w,x)P(\hat{y}_1|u,w,y_1)P(\hat{y}_2|y_2)$ and $\alpha_1 \in [0,1]$. The convex closure of $\mathcal{R}_i^{(1)}$ is achievable.
\end{proposition}

\begin{IEEEproof}
The cooperation strategy to derive the rate region $\mathcal{R}_i^{(1)}$ involves using Marton's coding scheme for the BC together with quantize-bin-and-forward at the receiver $Y_2$ first and then a combination of decode-and-forward and quantize-bin-and-forward at the receiver $Y_1$. At the receiver $Y_1$, the capacity of the conferencing link $C_{12}$ is divided into two parts. A fraction $\Bar{\alpha}_1C_{12}$ is devoted for decode-and-forward and the remaining $\alpha_1 C_{12}$ for quantize-bin-and-forward scheme. A detailed proof of Proposition \ref{th.ibc1} is given in Appendix \ref{app:achievability_1}.
\end{IEEEproof}

\begin{remark}\label{zetarep}
In the characterization of the rate region $\mathcal{R}_i^{(1)}$ given in Proposition \ref{th.ibc1}, one can replace $\zeta_1$ and $\zeta_2$ with $\Tilde{\zeta}_1$ and $\Tilde{\zeta_2}$, respectively, as given below:
\begin{align}
    %\begin{split}
\label{eq:zetarep_1}
        \Tilde{\zeta}_1&=\alpha_1 C_{12}-I(\hat{Y}_1;U,Y_1|V,W,Y_2)\\
\label{eq:zetarep_2}
        \Tilde{\zeta_2}&=C_{21}-I(\hat{Y}_2;Y_2|U,W,Y_1).
    %\end{split}
\end{align}
To see why this is the case, it is sufficient to note that in the case of
$\alpha_1 C_{12}<I(\hat{Y}_1;U,Y_1|V,W,Y_2)$, we can set $\hat{Y}_1 =
\varnothing$ and obtain the same rate region $\mathcal{R}_i^{(1)}$.
Similarly, in the case of $C_{21}<I(\hat{Y}_2;Y_2|U,W,Y_1)$, we can set 
$\hat{Y}_2 = \varnothing$ and obtain the same rate region. % $\mathcal{R}_i^{(1)}$.
\end{remark}
%\blacksquare

\newcounter{storeeqcounter_one}
\newcounter{tempeqcounter}

\begin{figure*}[!t]
%\normalsize
%\addtocounter{equation}{1}%
\setcounter{tempeqcounter}{\value{equation}} 
\begin{align}
\label{eq:ibc1eq}
R_0+R_1 &\le \min
    \left \{I(U,W;Y_1)+C_{21}-I(\hat{Y}_2;Y_2|U,W,Y_1), I(U,W;Y_1,\hat{Y}_2)\right \}\\
R_0+R_2 &\le I(V,W;Y_2)+C_{12}-I(\hat{Y}_1;U,Y_1|V,W,Y_2)\\
R_0+R_1+R_2 &\le \min \left \{I(U;Y_1|W)+C_{21}-I(\hat{Y}_2;Y_2|U,W,Y_1), I(U;Y_1,\hat{Y}_2|W)\right \}\nonumber\\
    &\qquad +I(V,W;Y_2)+C_{12}-I(\hat{Y}_1;U,Y_1|V,W,Y_2)-I(U;V|W)\\
R_0+R_1+R_2 &\le \min \left \{I(U,W;Y_1)+C_{21}-I(\hat{Y}_2;Y_2|U,W,Y_1), I(U,W;Y_1,\hat{Y}_2)\right \} \nonumber \\
	&\qquad +\min \left\{\begin{array}{l}
        {I(V;Y_2|W)+C_{12}-I(\hat{Y}_1;U,Y_1|V,W,Y_2),}\\
        {I(V;\hat{Y}_1,Y_2|W)}
    \end{array}\right\} -I(U;V|W)\\
2R_0+R_1+R_2 &\le \min
    \left \{I(U,W;Y_1)+C_{21}-I(\hat{Y}_2;Y_2|U,W,Y_1), I(U,W;Y_1,\hat{Y}_2)\right \}\nonumber\\
    &\qquad +I(V,W;Y_2)+C_{12}-I(\hat{Y}_1;U,Y_1|V,W,Y_2)-I(U;V|W)
\label{eq:ibc1eq_end}
\end{align}
\setcounter{equation}{\value{tempeqcounter}} % restore correct value
%\addtocounter{equation}{-1}%
\hrulefill
\vspace*{0pt}
\end{figure*}

\addtocounter{equation}{5}%
\setcounter{storeeqcounter_one}%
{\value{equation}}%

\begin{remark}
It is clear that an alternative achievable rate region can be derived by exchanging the cooperative protocol at the receivers in the achievability scheme of Proposition \ref{th.ibc1}, i.e., to apply quantize-bin-and-forward at the receiver $Y_1$ first and then a combination of decode-and-forward and quantize-bin-and-forward at the receiver $Y_2$.
\end{remark}

As discussed before, the parameter $\alpha_1$ in the achievable rate region $\mathcal{R}_i^{(1)}$ reflects a balance between the rates associated with the decode-and-forward and the quantize-bin-and-forward schemes at the receiver that applies both of them (i.e., $Y_1$). It turns out that the optimal value for the parameter $\alpha_1$ can be determined. This results in a simplified expression for the achievable rate region as given in the following theorem.

\begin{theorem}\label{alpha_1*}
%This gives the following simplified expression for $\mathcal{R}_i^{(1)}$:
For the two-user BC with conferencing decoders shown in Fig.~\ref{fig:BC_CD},
the rate region expressed in \eqref{eq:ibc1eq}-\eqref{eq:ibc1eq_end} 
over some joint PDF $P(u,v,w,x)P(\hat{y}_1|u,w,y_1)P(\hat{y}_2|y_2)$
is achievable.
This achievable rate region is obtained by setting the parameter $\alpha_1$ in $\mathcal{R}_i^{(1)}$ optimally as $\alpha_1^*$, given below:
\begin{equation}\label{eq:alpha*1}
    \alpha_1^*=\min \left\{\frac{I(\hat{Y}_1;U,Y_1|W,Y_2)}{C _{12}},1\right\}.
\end{equation}
\end{theorem}

\begin{IEEEproof}
%\textit{Proof of Corollary \ref{alpha_1*}:} 
First of all, $\zeta_1$ and $\zeta_2$ in the characterization of $\mathcal{R}_i^{(1)}$ can be respectively replaced by $\Tilde{\zeta}_1$ and $\Tilde{\zeta}_2$ given in (\ref{eq:zetarep_1})-(\ref{eq:zetarep_2}). Now, by setting $\alpha_1=\alpha_1^*$ in (\ref{eq:alpha*1}), one can easily derive the characterization (\ref{eq:ibc1eq})-(\ref{eq:ibc1eq_end}). To see that this choice is in fact optimal, it is sufficient to note that
\begin{align}
& \min 
\Big\{
I(V,W;Y_2)+\alpha_1 C_{12}-I(\hat{Y}_1;U,Y_1|V,W,Y_2)+\Bar{\alpha}_1C_{12}, \nonumber \\
& \qquad \quad \qquad I(V,W;\hat{Y}_1,Y_2)+\Bar{\alpha}_1C_{12} \Big\} \nonumber \\
& \qquad \le I(V,W;Y_2)+C_{12}-I(\hat{Y}_1;U,Y_1|V,W,Y_2) %\nonumber\\
\end{align}
and
\begin{align}
\min &
\left\{\begin{array}{l}
I(V;Y_2|W)+\alpha_1 C_{12}-I(\hat{Y}_1;U,Y_1|V,W,Y_2), \\
I(V;\hat{Y}_1,Y_2|W)
\end{array}\right\} \nonumber \\
\le & \min 
\left\{\begin{array}{l}
        {I(V;Y_2|W)+C_{12}-I(\hat{Y}_1;U,Y_1|V,W,Y_2),}\\
        {I(V;\hat{Y}_1,Y_2|W)}
\end{array}\right\}.
\end{align}
\end{IEEEproof}
%\blacksquare

%\subsection{Primitive Relays}

%\begin{figure}
%    \centering
%    \includegraphics[width=0.8\textwidth]{figs/Relay}
%    \caption{Primitive relay channel}
%    \label{fig:RC}
%\end{figure}

\begin{remark}
The achievable rate region $\mathcal{R}_i^{(1)}$ given in %Proposition \ref{th.ibc1} and 
Theorem \ref{alpha_1*} clearly includes both of the regions previously presented in \cite[Theorem 2]{Dabora} and \cite[Appendix B]{Gold} as subsets. Moreover, after some algebraic computation, one can show that $\mathcal{R}_i^{(1)}$, restricted over the distributions of the form $P(u,v,w,x)P(\hat{y}_1|w,y_1)P(\hat{y}_2|y_2)$, includes the achievable rate region of \cite[Theorem 3]{bross} as a subset. 
Therefore, $\mathcal{R}_i^{(1)}$ encompasses these previously results.
%We omit the computations here.
%\end{remark}
%\begin{remark}
The achievable rate region $\mathcal{R}_i^{(1)}$ is evaluated over a larger set of distributions, i.e., $P(u,v,w,x)P(\hat{y}_1|u,w,y_1)P(\hat{y}_2|y_2)$, as compared to the one in \cite[Theorem 3]{bross}. This improvement arises from the novel cooperation scheme at the first receiver, where it transmits a compressed version of the decoded signals U and W along with its received signal $Y_1$ to the second receiver. As discussed in \cite{FarsaniThesis}, this novel approach may lead to derivation of new capacity results for certain types of Gaussian channels.
\end{remark}

Let us now specialize the achievable rate region $\mathcal{R}_i^{(1)}$ for a
primitive relay channel \cite{Kimpr}. The primitive relay channel is a communication
scenario wherein a transmitter sends a message to a receiver with the help of
a relay node which has a digital relay link of given capacity to the receiver. %, as shown in Fig.~\ref{fig:RC}. 
This scenario is a special case of the two-user BC with conferencing decoders 
by setting $M_0=M_1=\varnothing$ and $C_{21}=0$ in Fig.~\ref{fig:BC_CD}.
Now, by setting $U = \varnothing$  and $V = X$ in %(\ref{eq:PRCar}), 
$\mathcal{R}_i^{(1)}$,
we obtain the following achievable rate for the primitive relay channel:
\begin{align}
& \max_{P(w,x)P(\hat{y}_1|w,y_1)} \min \Big\{
	{I(X;Y_2)+C_{12}-I(\hat{Y}_1;Y_1|X,W,Y_2),} \nonumber \\
& \qquad        {I(W;Y_1)+I(X;\hat{Y}_1,Y_2|W),} \nonumber \\
& \qquad       {I(W;Y_1)+I(X;Y_2|W)+C_{12}-I(\hat{Y}_1;Y_1|X,W,Y_2)}
\Big\}.
\end{align}
This is in fact the largest known achievable rate for the primitive relay channel as discussed in \cite{gamgoh}.

\begin{figure*}[!b]
\normalsize
\setcounter{tempeqcounter}{\value{equation}} 
\setcounter{equation}{38}%
\hrulefill
\begin{align}
\label{eq:ibc2eq}
%\begin{split}
    R_0+R_1 &\le I(U,W;Y_1)+C_{21}-I(\hat{Y}_2;Y_2|U,W,Y_1)\\
    R_0+R_2 &\le I(V,W;Y_2)\\
	R_0+R_1+R_2 &\le \min \left \{\begin{array}{l}
    	{I(U;Y_1|W)+C_{21}-I(\hat{Y}_2;Y_2|U,W,Y_1),}\\ {I(U;Y_1,\hat{Y}_2|W)}
    \end{array}\right\}
    +I(V,W;Y_2)-I(U;V|W)\\
    R_0+R_1+R_2 &\le I(U,W;Y_1)+C_{21}-I(\hat{Y}_2;Y_2|U,W,Y_1) \nonumber \\
    & \qquad +\min \left 
	\{\begin{array}{l}
    		{I(V;Y_2|W)+C_{12}-I(\hat{Y}_1;U,Y_1|V,W,Y_2),}\\
    		{I(V;\hat{Y}_1,Y_2|W)}
    	\end{array}\right\}
    -I(U;V|W)\\
    2R_0+R_1+R_2 &\le I(U,W;Y_1)+C_{21}-I(\hat{Y}_2;Y_2|U,W,Y_1)+I(V,W;Y_2)-I(U;V|W)
\label{eq:ibc2eq_end}
%\end{split}
\end{align}
\setcounter{equation}{\value{tempeqcounter}} % restore correct value
%\addtocounter{equation}{-1}%
\vspace*{0pt}
\end{figure*}

Next, we present a second achievable rate region %for the two-user BC with bidirectional cooperation 
based on a combination of
decode-and-forward and quantize-bin-and-forward at one receiver first and then
quantize-bin-and-forward at the other receiver as the cooperation strategy. 
This strategy may be advantageous when part of the message is already 
decodable at the first receiver before conferencing.

\begin{proposition}\label{th.ibc2}
Consider the two-user BC with conferencing decoders shown in Fig.~\ref{fig:BC_CD}. Let $\mathcal{R}_i^{(2)}$ denote the set of all rate triples $(R_0,R_1,R_2)\in \mathbb {R}_+^3$ such that
\begin{align}
%\begin{split}
\label{eq:ibc2}
    R_0+R_1 &\le \min
    \Big\{I(U,W;Y_1)+\eta_2+\Bar{\alpha}_2C_{21}, \nonumber \\ 
& \qquad I(U,W;Y_1,\hat{Y}_2)+\Bar{\alpha}_2C_{21}\Big \}\\
    R_0+R_2 &\le I(V,W;Y_2)\\
    R_0+R_1+R_2 &\le \min \Big \{I(U;Y_1|W) +\eta_2, I(U;Y_1,\hat{Y}_2|W)\Big\} \nonumber \\ 
& \ +I(V,W;Y_2)-I(U;V|W)\\
    R_0+R_1+R_2 &\le \min \Big \{I(U,W;Y_1)+\eta_2+\Bar{\alpha}_2C_{21}, \nonumber \\
& \qquad \qquad I(U,W;Y_1,\hat{Y}_2)+\Bar{\alpha}_2C_{21}\Big \} \nonumber \\
& \ + \min \left \{I(V;Y_2|W)+\eta_1,I(V;\hat{Y}_1,Y_2|W)\right\} \nonumber \\
& \ -I(U;V|W)\\
2R_0+R_1+R_2 &\le \min
    \Big \{I(U,W;Y_1)+\eta_2+\Bar{\alpha}_2C_{21}, \nonumber \\
& \qquad \qquad I(U,W;Y_1,\hat{Y}_2)+\Bar{\alpha}_2C_{21}\Big \} \nonumber \\
    &\ +I(V,W;Y_2)-I(U;V|W)
\label{eq:ibc2_end}
%\end{split}
\end{align}
where
\begin{align}
    \eta_1&=\{C_{12}-I(\hat{Y}_1;U,Y_1|V,W,Y_2)\}^+\\
    \eta_2&=\{\alpha_2 C_{21}-I(\hat{Y}_2;Y_2|U,W,Y_1)\}^+
\end{align}
for some joint PDFs $P(u,v,w,x)P(\hat{y}_1|u,w,y_1)P(\hat{y}_2|w,y_2)$ and $\alpha_2 \in [0,1]$. The convex closure of $\mathcal{R}_i^{(2)}$ is achievable.
\end{proposition}

\begin{IEEEproof}
The cooperation strategy to achieve the rate region $\mathcal{R}_i^{(2)}$ is to apply Marton's coding with a combination of decode-and-forward and quantize-bin-and-forward at the receiver $Y_2$ first, then quantize-bin-and-forward at the receiver $Y_1$. At the receiver $Y_2$, the capacity of the conferencing link $C_{21}$ is divided into two parts. A fraction $\Bar{\alpha}_2C_{21}$ is devoted for decode-and-forward and the remaining $\alpha_2 C_{21}$ is devoted for quantize-bin-and-forward scheme. The details of the proof are in
	%of Proposition \ref{th.ibc2} can be found in 
	Appendix \ref{app:achievability_2}. 
\end{IEEEproof}

\begin{remark}\label{etarep}
In the characterization of the rate region (\ref{eq:ibc2})-(\ref{eq:ibc2_end}), one can replace $\eta_1$ and $\eta_2$ with $\Tilde{\eta}_1$ and $\Tilde{\eta_2}$, respectively, which are given as follows:
\begin{align}\label{eq:etarep}
%    \begin{split}
        \Tilde{\eta}_1&=C_{12}-I(\hat{Y}_1;U,Y_1|V,W,Y_2)\\
        \Tilde{\eta_2}&=\alpha_2 C_{21}-I(\hat{Y}_2;Y_2|U,W,Y_1).
%    \end{split}
\end{align}
%\begin{IEEEproof}
%\textit{Proof of Remark \ref{etarep}:} 
To see why this is the case, it is sufficient to note that in the case of
$C_{12}<I(\hat{Y}_1;U,Y_1|V,W,Y_2)$, we can set $\hat{Y}_1 = 
\varnothing$ and obtain the same rate region $\mathcal{R}_i^{(2)}$.
Similarly, in the case of $\alpha_2 C_{21}-I(\hat{Y}_2;Y_2|U,W,Y_1)$, we can set 
$\hat{Y}_2 = \varnothing$ and obtain the same rate region. 
\end{remark}
%To see why this is the case, it is sufficient to note that the case of $C_{12}<I(\hat{Y}_1;U,Y_1|V,W,Y_2)$ in the characterization of the rate region $\mathcal{R}_i^{(2)}$ is equivalent to the case where we set $\hat{Y}_1$ equal to $\varnothing$. Similarly, the case of $\alpha_2 C_{21}-I(\hat{Y}_2;Y_2|U,W,Y_1)$ is equivalent to the case where we set $\hat{Y}_2$ equal to $\varnothing$.
%\end{IEEEproof}
%\blacksquare
%\\

\begin{remark}
One can obtain an alternative achievable rate region by exchanging the cooperative protocol at the receivers in the achievability scheme of Proposition \ref{th.ibc2}, i.e., to apply a combination of decode-and-forward and quantize-bin-and-forward at the receiver $Y_1$ first and then quantize-bin-and-forward at the receiver $Y_2$.
\end{remark}

Similar to the region $\mathcal{R}_i^{(1)}$, we can derive the optimal value of the parameter $\alpha_2$ for the region $\mathcal{R}_i^{(2)}$. % as given in the following.

\begin{theorem}\label{alpha_2*}
For the two-user BC with conferencing decoders shown in Fig.~\ref{fig:BC_CD}, 
the rate region expressed in \eqref{eq:ibc2eq}-\eqref{eq:ibc2eq_end} over
some joint PDFs $P(u,v,w,x)P(\hat{y}_1|u,w,y_1)P(\hat{y}_2|w,y_2)$ is 
achievable. This achievable rate region is obtained by setting 
the parameter $\alpha_2$ in $\mathcal{R}_i^{(2)}$ optimally as $\alpha_2^*$, 
given below:
\newcounter{storeeqcounter_two}
\addtocounter{equation}{5}%
\setcounter{storeeqcounter_two}%
{\value{equation}}%
\begin{equation}\label{eq:alpha_2*}
    \alpha_2^* %=\min \left\{\frac{I(\hat{Y}_2;U,Y_2|W,Y_1)}{C _{21}},1\right\}
    =\min \left\{\frac{I(\hat{Y}_2;Y_2|W,Y_1)}{C _{21}},1\right\}.
\end{equation}
\end{theorem}

%Moreover, $\mathcal{R}_i^{(2)}$ is equivalent to the following simplified region:

\begin{IEEEproof}
%\textit{Proof of Corollary \ref{alpha_2*}:} 
	First of all, $\eta_1$ and $\eta_2$ in the characterization (\ref{eq:ibc2})-(\ref{eq:ibc2_end}) are respectively replaced by $\Tilde{\eta}_1$ and $\Tilde{\eta}_2$ given in (\ref{eq:etarep}). Now, by setting $\alpha_2=\alpha_2^*$ in (\ref{eq:alpha_2*}), we easily obtain the characterization (\ref{eq:ibc2eq})-(\ref{eq:ibc2eq_end}). To prove that this choice is in fact optimal, it is sufficient to consider the following inequalities:
\begin{align}
& \min \Big \{I(U,W;Y_1)+\alpha_2 C_{21}-I(\hat{Y}_2;Y_2|U,W,Y_1)+\Bar{\alpha}_2C_{21}, \nonumber \\
& \qquad \qquad I(U,W;Y_1,\hat{Y}_2)+\Bar{\alpha}_2C_{21}\Big \} \nonumber \\
& \qquad \le I(U,W;Y_1)+C_{21}-I(\hat{Y}_2;Y_2|U,W,Y_1)
\end{align}
and
\begin{align}
& \min 
\left \{\begin{array}{l}
I(U;Y_1|W)+\alpha_2 C_{21}-I(\hat{Y}_2;Y_2|U,W,Y_1), \\I(U;Y_1,\hat{Y}_2|W)
\end{array}\right\} \nonumber \\
& \le \min 
\left \{\begin{array}{l}
{I(U;Y_1|W)+C_{21}-I(\hat{Y}_2;Y_2|U,W,Y_1),}\\{I(U;Y_1,\hat{Y}_2|W)}
\end{array}\right\}
%    \end{split}
\end{align}
and by noting that $I(\hat{Y}_2;U,Y_2|W,Y_1) =I(\hat{Y}_2;Y_2|W,Y_1)$.
%\\
%\blacksquare
%\\
\end{IEEEproof}

In general, it is not easy to compare the two regions $\mathcal{R}_i^{(1)}$ and $\mathcal{R}_i^{(2)}$ numerically. An advantage of the region $\mathcal{R}_i^{(2)}$ over the region $\mathcal{R}_i^{(1)}$ is that its feasible set of distributions is a larger set. %This is the justification that we present it in this paper separately. 
We further note that by allowing a fourth auxiliary random variable, it is possible to design an achievability scheme that includes both of $\mathcal{R}_i^{(1)}$ and $\mathcal{R}_i^{(2)}$ as special cases. However, evaluating the resulting rate region would be complex. %{\color{blue} In fact, these achievable rate regions can be derived from the general inner bounds of \cite{chung}, but the general expressions of \cite{chung} are not always easy to simplify.}

\subsection{Capacity Regions for Specific Channels}

The achievability results in the previous section may be further improved, e.g., by 
applying multiple rounds of cooperation at the receivers. However, in what follows, we demonstrate that the rate region $\mathcal{R}_i^{(1)}$ is already sufficient to establish new capacity results for several particular two-user BCs with conferencing decoders (as well as approximate capacity results for the Gaussian channel as seen in the next section). The first such capacity result is for a class of semi-deterministic BCs with degraded message sets. 
Note that the determinism here is more general than that of semi-deterministic BC in \cite{Gold}. 
% \footnote{Note that the definition of determinism in this paper is different from that of semi-deterministic BC in \cite{Gold}.} 

\begin{theorem}\label{th.dbc}
	Consider the two-user BC with degraded message sets, with a common message for both users and a private message for the first user, and bidirectional conferencing decoders. If the BC satisfies certain semi-deterministic property in the sense that $Y_2=f(X,Y_1)$, then its capacity region is given by the closure of the convex hull of all $(R_0,R_1)$ satisfying
\begin{align}\label{eq:dbc}
    R_0 &\le I(V;Y_2)+C_{12}\\
    R_0+R_1 &\le I(X;Y_1)+C_{21}\\
    R_0+R_1 &\le I(X;Y_1|V)+I(V;Y_2)+C_{12}+C_{21}\\
    R_0+R_1 &\le I(X;Y_1,Y_2|V)+I(V;Y_2)+C_{12} \label{eq:BC_degraded_Y1Y2} \\
    R_0+R_1 &\le I(X;Y_1,Y_2)
\end{align}
for some joint PDFs $P(v,x)$,
where $V \rightarrow X \rightarrow Y_1,Y_2$ forms a Markov chain and $|\mathcal{V}|\le |\mathcal{X}|+2$.
\end{theorem}

\begin{IEEEproof}
%\textit{Proof of Theorem \ref{th.dbc}:} 
The achievability is derived by setting $U = X$, $W = V$, $\hat{Y}_1=\varnothing$, and $\hat{Y}_2 = Y_2$ in $\mathcal{R}_i^{(1)}$. The converse proof is readily given by $\mathcal{R}_o$. %\blacksquare
\end{IEEEproof}

Note that $(Y_1, Y_2)$ appears jointly in the first mutual information expression in
\eqref{eq:BC_degraded_Y1Y2}. Thus, the capacity region of this BC depends on the 
joint distribution $p(y_1,y_2|x)$ and not just on the marginals, in contrast to 
BC without conferencing decoders. %This is a first example of capacity region for a BC that depends on the joint transitional probability distribution function of the channel.

\begin{remark}
By setting $R_0=0$ and $C_{12}=0$, the result of Theorem \ref{th.dbc} is reduced to the capacity of deterministic primitive relay channel derived in \cite{Kimd}. In this case, the capacity is given by
	%{\color{red}
	$\max_{P(x)} \min\{I(X;Y_1)+C_{21}, I(X;Y_1,Y_2)\}$. Note that for the case of $R_0=0$ and $C_{12}=0$, it is optimal to set $V = \varnothing$. In fact, the cut-set bound is achievable. This capacity is achieved by quantize-bin-and-forward as the relay strategy.
\end{remark}

We present two interesting examples of this class of BC. %semi-deterministic BC with degraded message set. 

\begin{example}
\label{ex:DMC_1}
Consider the following binary channel,
\begin{align}
\begin{split}
\begin{cases}
    Y_1=X \oplus Z \\
    Y_2=Z
\end{cases}
\end{split}
\end{align}
where $Z$ is a binary noise and $\oplus$ is the XOR operator. For this channel, we have $Y_2=X \oplus Y_1$, so the channel has the desired deterministic property. The capacity region of this BC with degraded message sets is given by all $(R_0,R_1)$ satisfying 
\begin{align}\label{eq:bex1}
%\begin{split}
    R_0 &\le C_{12}\\
    R_0+R_1 &\le I(X;Y_1)+C_{21}\\
    R_0+R_1 &\le H(X)
%\end{split}
\end{align}
for some $P(x)$. 
\end{example}

	In this example, the receiver $Y_2$ does not receive information from the transmitter directly. Instead, it observes the noise $Z$ and relays the noise to the receiver $Y_1$ by sending a compressed version of it through the digital link $C_{21}$. The receiver $Y_1$ then decodes both messages and forwards the common message to the receiver $Y_2$ through the digital link $C_{12}$. Therefore, the receiver $Y_2$ can still receive information at a positive rate.

\begin{example}
\label{ex:DMC_2}
Consider the following binary channel, 
\begin{align}
\begin{cases}
    Y_1=Z \\
    Y_2=X \oplus Z
\end{cases}
\end{align}
where again $Z$ is a binary noise. The capacity region of this BC with degraded message sets is the closure of the convex hull of all $(R_0,R_1)$ satisfying 
\begin{align}\label{eq:bex2}
    R_0 &\le I(V;Y_2)+C_{12}\\
    R_0+R_1 &\le C_{21}\\
    R_0+R_1 &\le H(X|V)+I(V;Y_2)+C_{12} \label{eq:bex2_3th} \\
    R_0+R_1 &\le H(X)
\end{align}
for some $P(v,x)$, where $V \rightarrow X \rightarrow Y_1,Y_2$ forms a Markov chain and $|\mathcal{V}|\le 3$. %|\mathcal{X}|+1$.
\end{example}

 \begin{figure}
    \centering
    \includegraphics[width=\columnwidth]{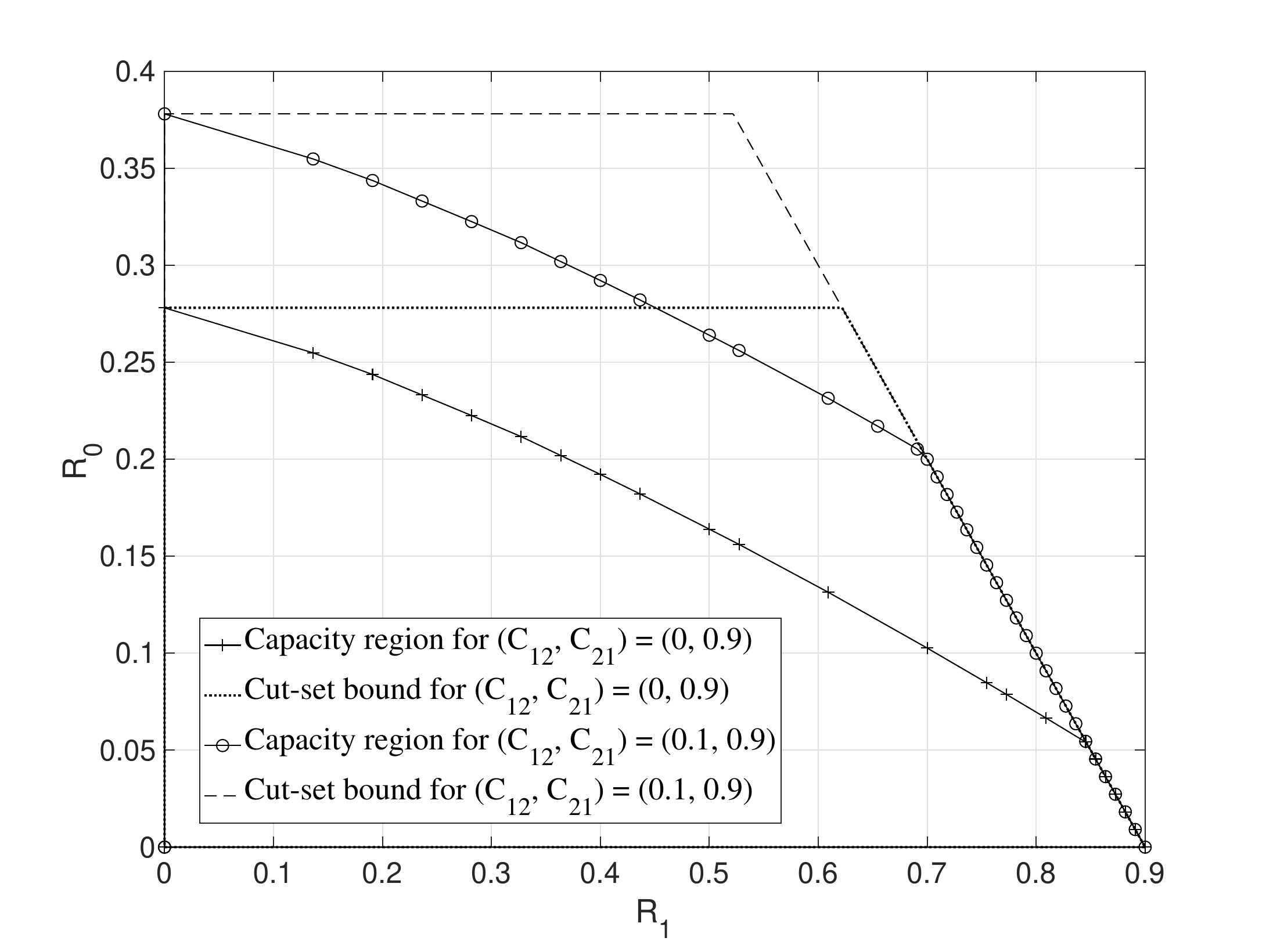}
    \caption{Capacity region and cut-set bound for the BC in Example \ref{ex:DMC_2}, where $Z$ is a binary noise with ${P(Z=0) = 0.2}$.}
    \label{ex:DMC_2-plots}
\end{figure}

Unlike the previous example, for which the capacity region coincides with the cut-set bound, the capacity region for this BC is characterized using an auxiliary variable $V$. %The optimal choice of $V$ depends on the values of $C_{12}$ and $C_{21}$. 
In Fig.~\ref{ex:DMC_2-plots}, we plot the capacity region along with the cut-set bound for the channel in Example \ref{ex:DMC_2} with $P(Z=0)=0.2$ for two cases: $(C_{12},C_{21}) = (0, 0.9)$ and $(C_{12},C_{21}) = (0.1, 0.9)$, by searching over all possible $P(v,x)$. As can be seen, the capacity region lies strictly below the cut-set bound in this example. Moreover, $C_{12}$ strictly improves the capacity region, although $Y_1$ only observes the noise. In fact, the constraint \eqref{eq:bex2_3th} (which arises from \eqref{eq:BC_degraded_Y1Y2}) plays a crucial role in characterizing the capacity region. Without (\ref{eq:bex2_3th}), the remaining constraints simply reduce to the cut-set bound.

For this channel, the optimal coding strategy is as follows. Note that the receiver $Y_1$ (which is supposed to decode both common and private messages) does not receive any information from the transmitter directly. The optimal cooperation strategy is that the receiver $Y_2$ first sends a compressed version of its received signal to the receiver $Y_1$ through the digital link $C_{21}$. Next, the receiver $Y_1$ decodes both messages using the information received (and its own signal which is in fact the channel noise) and then performs decode-and-forward of the common message through the link $C_{12}$ to the receiver $Y_2$ (using the variable $V$). Lastly, the receiver $Y_2$ decodes the common message using also its own signal. 

It is worthwhile to make the following observation. In Example \ref{ex:DMC_2}, as the receiver $Y_1$ observes the channel noise only, one might think that a cooperative scheme in which $Y_1$ applies quantize-bin-and-forward and $Y_2$ applies decode-and-forward is the right strategy. However, it turns out that such a scheme is not optimal when the message to $Y_2$ is a degraded version of the message to $Y_1$. This is in contrast to Example \ref{ex:DMC_1}, in which the message set degrades in the opposite direction and performing quantize-bin-and-forward at the receiver that observes the noise only is the capacity-achieving strategy.

%Note that the capacity region of Example \ref{ex:DMC_1} coincides with the cut-set bound. But the capacity region of Example \ref{ex:DMC_2} is strictly below the cut-set bound.

As far as we know, the above results are the first cases where a combination of quantize-bin-and-forward and decode-and-forward strategies yields an \textit{optimal bidirectional cooperation protocol}. There is no previously known capacity result in the literature for channels with bidirectional cooperation between the receivers. Moreover, these results demonstrate that a one-round cooperation scheme is sufficient for achieving capacity for this class of channels. 

It is also possible to provide a capacity result for a new class of more capable semi-deterministic BCs with both common and private messages and one-sided conferencing. %The result is given in the next theorem.

\begin{theorem}\label{th.dmbc}
Consider the two-user BC with both common and private messages and with certain semi-deterministic property in the sense that $Y_2=f(X,Y_1)$. Moreover, assume that the channel is
more-capable, i.e., $I(X;Y_2)\le I(X;Y_1)$ for every input distribution
$P(x)$. For the channel with one-sided conferencing, i.e., 
with a digital relay link of finite capacity $C_{21}$ from the decoder at 
$Y_2$ to the decoder at $Y_1$, the capacity region is given by the closure 
of the convex hull of all $(R_0,R_1,R_2)$ satisfying 
\begin{align}\label{eq:dmbc}
    R_0+R_2 &\le I(V;Y_2)\\
    R_0+R_1+R_2 &\le I(X;Y_1)+C_{21}\\
    R_0+R_1+R_2 &\le I(X;Y_1|V)+I(V;Y_2)+C_{21}\\
    R_0+R_1+R_2 &\le I(X;Y_1,Y_2|V)+I(V;Y_2)\\
    R_0+R_1+R_2 &\le I(X;Y_1,Y_2)
\end{align}
for some joint PDFs $P(v,x)$,
where $V \rightarrow X \rightarrow Y_1,Y_2$ forms a Markov chain and $|\mathcal{V}|\le |\mathcal{X}|+2$.
\end{theorem}

\begin{IEEEproof}
%\textit{Proof of Theorem \ref{th.dmbc}:} 
The achievability is derived by setting $U = X$, $W = V$,  $\hat{Y}_1=\varnothing$, and $\hat{Y}_2 = Y_2$ in $\mathcal{R}_i^{(1)}$. %The converse is given by $\mathcal{R}_o$. %$\blacksquare$
Note that the more capable condition and one-sided conferencing imply that the private message for user 2 can always be decoded by user 1.

For the converse, note that due to the Markov chain condition $U \rightarrow X \rightarrow Y_1,Y_2$, the more-capable condition implies% 
\begin{align}
    I(X;Y_2|U) \le I(X;Y_1|U) \quad \forall P(u,x).
\end{align}
Therefore, the constraint \eqref{eq:outer_RsumU_1} from the outer bound $\mathcal{R}_o$ with $C_{12} = 0$ becomes 
\begin{align}
R_0 + R_1 + R_2 &\leq I(X;Y_2|U) + I(U;Y_1) + C_{21} \\
&\leq I(X;Y_1|U) + I(U;Y_1) + C_{21} \\
& = I(X;Y_1) + C_{21},
\end{align}
where the last step is again due to the Markov chain.
The remaining constraints are directly derived from $\mathcal{R}_o$ by setting $C_{12} = 0$.
The resulting outer bound coincides with  $\mathcal{R}_i^{(1)}$.
\end{IEEEproof}

\begin{remark}
The structure of the novel outer bound $\mathcal{R}_o$ given in Theorem \ref{th.obc}, in particular, with respect to the constraints \eqref{eq:outer_RsumV_2} and \eqref{eq:outer_RsumU_2}, is crucial for deriving the capacity results in Theorems \ref{th.dbc} and \ref{th.dmbc}. 
Specific combinations of conditions, such as semi-determinism, degraded message sets, and the more capable channel, make it possible for the converse
%\eqref{eq:outer_RsumV_2} and \eqref{eq:outer_RsumU_2} 
	to be tight.
%These results carry over to the Gaussian BCs.
%In the next section, %\cite{faryuj}, we show that the bounds can also be used to derive capacity results for certain classes of Gaussian BCs. 
\end{remark}

\section{Gaussian BC with Conferencing Decoders}
\label{sec:gaussian}

We now study the Gaussian BC with conferencing decoders and with correlated noises as shown in Fig.~\ref{fig:Gaussian}. The channel model is defined as follows:
\begin{align}\label{eq:GBC}
\begin{cases}
    Y_1=aX+Z_1\\
    Y_2=bX+Z_2
\end{cases}
\end{align}
where $Z_1$ and $Z_2$ are correlated Gaussian noises with zero mean and unit variance and correlation coefficient $\lambda$, i.e., $\mathbb{E}[Z_1Z_2]=\lambda$, $X$ is the input signal with $\mathbb{E}[X^2]\le P$, and $a$, and $b$ are the real-valued channel gains. %%Without lack of generality, throughout the document, we assume that $|a|\ge |b|$ where $|.|$ stands for absolute value.
Throughout this section, the transmit power constraint is assumed to be \emph{strictly} positive, i.e., $P > 0$.
The Gaussian BC is equipped with a conferencing link of capacity $C_{12}$ from the decoder at $Y_1$ to the decoder at $Y_2$, and a conference link of capacity $C_{21}$ from the decoder at $Y_2$ to the decoder at $Y_1$. 

We show that the bounds derived in earlier part of the paper are useful for deriving capacity and approximate capacity results for various types of Gaussian BC with conferencing decoders.

\begin{figure}
    \centering
    \includegraphics[width=0.5\textwidth]{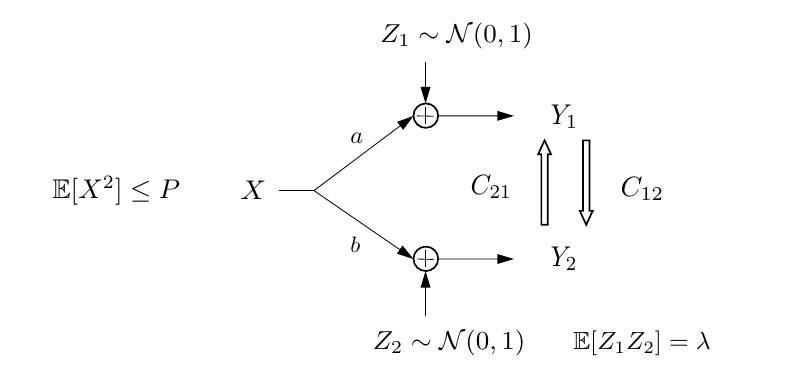}
    \caption{Two-user Gaussian BC with power constraint $P>0$ and noise correlation $\lambda$, and with conferencing links of capacities $C_{12}$ and $C_{21}$. }
    \label{fig:Gaussian}
\end{figure}

\subsection{Converse}

We begin with the outer bound.  By taking into consideration the input
power constraint $\mathbb{E}[X^2]\le P$, %the outer bound is also valid for the Gaussian channel (\ref{eq:GBC}). 
%and by using the entropy power inequality, 
the outer bound in Theorem \ref{th.obc} can be optimized over its auxiliary variables $U$ and $V$ for the Gaussian channel. 
This gives 
an explicit characterization of the mutual information terms in the
outer bound as derived below.

\begin{theorem}\label{th.ogbc}
Consider the Gaussian BC (\ref{eq:GBC}) with conferencing decoders. Assume that $|a|\ge|b|$. % and $\lambda \neq \frac{b}{a}$. 
Let $\mathcal{R}_o^G$ denote the set of all rate triples $(R_0,R_1,R_2)$ such that for some $\alpha, \beta \in [0,1]$
\begin{align}
\label{eq:ogbc}
   R_0+R_1 &\le \psi\left(\frac{(1-\alpha)a^2P}{\alpha a^2P+1}\right)+C_{21}\\
   R_1 &\le \Psi_2+\psi\left(\frac{(1-\beta)b^2P}{\beta b^2P+1}\right) \label{eq:ogbc_2} \\
   R_0+R_2 &\le \psi\left(\frac{(1-\beta)b^2P}{\beta b^2P+1}\right)+C_{12} \label{eq:ogbc_3} \\
   R_2 &\le \Psi_1+\psi\left(\frac{(1-\alpha)b^2P}{\alpha b^2P+1}\right) \label{eq:ogbc_4} \\
   R_0+R_1+R_2 &\le \psi\left(\beta a^2P\right)+\psi\left(\frac{(1-\beta)b^2P}{\beta b^2P+1}\right) \nonumber \\ & \qquad \qquad \qquad \qquad %\qquad 
	+C_{12}+C_{21} \label{eq:ogbc_5} \\
   R_0+R_1+R_2 &\le \Psi_2+\psi\left(\frac{(1-\beta)b^2P}{\beta b^2P+1}\right)+C_{12} \label{eq:ogbc_6} \\
   R_0+R_1+R_2 &\le \Psi_1+\psi\left(\frac{(1-\alpha)a^2P}{\alpha a^2P+1}\right)+C_{21} \label{eq:ogbc_7} \\
   R_0+R_1+R_2 &\le \psi\left(\left(\frac{a^2+b^2-2\lambda ab}{1-\lambda^2}\right)P\right),
\label{eq:ogbc_8}
\end{align}
where %$\psi(x) = \frac{1}{2} \log(1+x)$ and
\begin{align}
\Psi_1&=\psi\left(\alpha\left(\frac{a^2+b^2-2\lambda ab}{1-\lambda^2}\right)P\right)\\
\Psi_2&=\psi\left(\beta\left(\frac{a^2+b^2-2\lambda ab}{1-\lambda^2}\right)P\right).
\end{align}
Then $\mathcal{R}_o^G$ constitutes an outer bound on the capacity region of the BC.
\end{theorem}

\begin{IEEEproof}
The outer bound is based on applying the entropy power inequality to the outer
bound derived in Theorem \ref{th.obc}. The details are given in Appendix \ref{app:outer_bound_gbc}. %\qed % $\blacksquare$
\end{IEEEproof}

\begin{remark}\label{degraded1}
For the Gaussian BC (\ref{eq:GBC}) with $|a|\ge|b|$, when $\lambda=\frac{b}{a}$, 
we have
%the channel is degraded and we have
\begin{align}
    \frac{a^2+b^2-2\lambda ab}{1-\lambda^2}= a^2. %a^2+\frac{(\lambda a-b)^2}{1-\lambda^2}=a^2.
\label{eq:condition_lambda_ab}
\end{align}
Note that in the case of
$\lambda=1$ (with $a=b$), or $\lambda=-1$ (with $a=-b$), the above equality 
is to be interpreted in the limiting sense. 
Further, when $\lambda = \frac{b}{a}$, the BC is physically degraded, for which 
the outer bound $\mathcal{R}_o^G$ is achievable. 
%(This case of $\lambda = \frac{b}{a}$ is further discussed in Remark \ref{degraded2}.)
(See Remark \ref{degraded2}.)
\end{remark}

We now discuss the BC with perfectly correlated noises, i.e., $\lambda=1$ or $-1$, but with the assumption that $\lambda \neq \frac{b}{a}$.  

\subsection{Capacity Regions for Specific Channels}

\begin{theorem}\label{th.FCgbc1}
	Consider the two-user Gaussian BC (\ref{eq:GBC}) with degraded message sets (i.e., a common message for both receivers at rate $R_0$ and a private message for the first receiver at rate $R_1$) and bidirectional conferencing decoders. 
Assume $P>0$. When the noises are perfectly correlated, i.e., $|\lambda|=1$, 
but $\lambda \neq \frac{b}{a}$, the capacity region of the channel is given by
the closure of the convex hull of all $(R_0,R_1)$ satisfying
\begin{align}
\label{eq:FCgbc1} 
R_0 &\le \psi\left(\frac{(1-\beta)b^2P}{\beta b^2P+1}\right)+C_{12}\\
    R_0+R_1 &\le \psi\left(a^2P\right)+C_{21}\\
    R_0+R_1 &\le \psi\left(\beta a^2P\right)+\psi\left(\frac{(1-\beta)b^2P}{\beta b^2P+1}\right)+C_{12}+C_{21}
\label{eq:FCgbc1_3} 
\end{align}
for some $\beta\in[0,1]$.
\end{theorem}

\begin{IEEEproof}
%\textit{Proof:} 
The achievability is derived from $\mathcal{R}_i^{(1)}$ given in Proposition \ref{th.ibc1} by setting $X = U = W+\Bar{W}$ and $V = \varnothing$, and $\hat{Y}_2 = Y_2+\hat{Z}_2$,  where $W$ and $\Bar{W}$ are two independent Gaussian variables with zero means and variances $(1-\beta) P$ and $\beta P$, respectively, and $\hat{Z}_2$ is a Gaussian variable (independent of all other variables) with zero mean and variance $\hat{\sigma}^2$. Note that by this choice of random variables, when $|\lambda|=1$, we have
\begin{align}
    I(\hat{Y}_2;Y_2|W,U,Y_1)=\psi \left(\frac{1-\lambda^2}{\hat{\sigma}^2}\right)=0
\end{align}
Moreover,
\begin{align}
    I(X;Y_1,\hat{Y}_2)&=\psi\left(\left(\frac{a^2+b^2-2\lambda ab+\hat{\sigma}^2a^2}{1-\lambda^2+\hat{\sigma}^2}\right) P\right) \nonumber \\
    &=\psi\left(\left(\frac{a^2+b^2-2\lambda ab+\hat{\sigma}^2a^2}{\hat{\sigma}^2}\right)P\right)
\end{align}
and 
\begin{align}
    I(X;Y_1,\hat{Y_2}|W)&=\psi\left(\beta\left(\frac{a^2+b^2-2\lambda ab+\hat{\sigma}^2a^2}{1-\lambda^2+\hat{\sigma}^2}\right)P\right) \nonumber \\
    &=\psi\left(\beta\left(\frac{a^2+b^2-2\lambda ab+\hat{\sigma}^2a^2}{\hat{\sigma}^2}\right)P\right).
\end{align}
Note that if $|\lambda|=1$ but $\lambda \neq \frac{b}{a}$, then
$a^2+b^2-2\lambda ab+\hat{\sigma}^2a^2 \neq 0$.
Therefore, by letting $\hat{\sigma}^2\rightarrow 0$, one can make the above
two mutual information terms arbitrarily large, so they would not be
in effect in the characterization of the achievable rate region. % (\ref{eq:ibc}). 
For the case of $|a|\ge|b|$, the converse proof is readily given by $\mathcal{R}_o^G$ in Theorem \ref{th.ogbc}. %Note that the constraint (13) is stricter than (24). 
For the case of $|a|<|b|$, the achievable rate region (\ref{eq:FCgbc1})-(\ref{eq:FCgbc1_3}) is maximized for $\beta =0$ and it coincides with the cut-set bound.    %$\blacksquare$
\end{IEEEproof}

%important from two viewpoints. First, it is the first capacity result for a two-user Gaussian BC with bidirectional cooperation between receivers (all previously capacity results are regarding one-sided cooperation). Second, this result 
The capacity result in Theorem \ref{th.FCgbc1} is 
among the rare cases in network information theory for which the quantize-bin-and-forward strategy contributes to achieving capacity.

\begin{example}
\label{ex:GBC_1}
A special case of the Gaussian BC with perfectly correlated noises is the following channel:
\begin{align}\label{eq:GweiBC}
\begin{cases}
    Y_1=X+Z\\
    Y_2=X-Z
\end{cases}
\end{align}
where $Z$ is a zero-mean unit-variance Gaussian noise. In this scenario, the two receivers see exactly the same noise but with a different sign. The capacity region of this channel with degraded message sets can be obtained from Theorem \ref{th.FCgbc1} as: 
\begin{align}\label{eq:gex1}
     R_0 &\le \psi(P)+C_{12}\\
    R_0+R_1 &\le \psi(P)+C_{21}.
\end{align}
\end{example}

In fact for this channel, the cut-set bound is achievable. 
To achieve this capacity, the second receiver applies quantize-bin-and-forward and 
the first receiver applies decode-and-forward as the cooperation protocol. 

Observe that in the case of common message only, even as $P \rightarrow 0$, it is possible to achieve a strictly positive rate of $R_0 = \min\{ C_{12},C_{21} \}$.
This makes concrete the statement made in the introduction of this paper, namely, that it is possible to use an infinitesimally small amount of power to transmit information at a strictly positive rate.
Note that the argument presented in the proof of Theorem \ref{th.FCgbc1} is only valid for strictly positive values of input power $P$.
We summarize this observation in the following remark.

%As already mentioned, for the special case of $R_0=0$ and $C_{12}=0$, this result recovers the capacity of semi-deterministic relay channel \cite [Example 1]{Kimd}.

\begin{remark}
	The capacity characterization given in Theorem \ref{th.FCgbc1} gives rise to the following interesting observation. Even with a very small (yet positive) amount of input power $P$, one can transmit information over the channel at a rate as high as the capacities of the conferencing links, if the noises are perfectly correlated, i.e., $|\lambda|=1$, (assuming $\lambda \neq \frac{b}{a}$). In fact, as $P \rightarrow 0$, the capacity region of the Gaussian BC with degraded message sets is as follows:
\begin{align}
     R_0 &\le \epsilon_1(P)+C_{12}\\
    R_0+R_1 &\le \epsilon_2(P)+C_{21}
\end{align}
for some $\epsilon_1 (P)$ and $\epsilon_2(P)$, which go to zero
%where $\epsilon_1 (P)$, $\epsilon_2(P) \rightarrow 0$ 
as $P\rightarrow 0$. 
\end{remark}

Next, we draw a connection between the BC with conferencing decoders and the relay channel.

\begin{remark}
For the special case of a primitive relay channel with perfectly correlated noises, i.e., $|\lambda|=1$, where $Y_2$ acts as a relay for $Y_1$, $C_{12}=0$, and there is only a private message for the first receiver, the capacity result given in Theorem \ref{th.FCgbc1} is reduced to $R_1 \le \psi\left(a^2P\right)+C_{21}$. In this case, using the quantize-bin-and-forward strategy, one can achieve the cut-set bound. Interestingly, the capacity does not depend on $b$ at all (as long as $\lambda \neq \frac{b}{a}$). 
\end{remark}

This result is an example of a semi-deterministic primitive relay channel, because the relay observation $Y_2=b X + Z_2$ is a deterministic function of the input $X$ and the receiver observation $Y_1=aX+Z_1$, when $Z_1$ and $Z_2$ are perfectly correlated (i.e., $|\lambda|=1$) and $\lambda \neq \frac{b}{a}$. In this case, the cut-set bound is achievable \cite{Kimd}.

The following two examples further illustrate Theorem \ref{th.FCgbc1}.

%Let us provide some examples.

\begin{example}
\label{ex:GBC_2}
Consider the following Gaussian BC:
\begin{align}\label{eq:GBCex2}
\begin{cases}
    Y_1=X+Z\\
    Y_2=Z
\end{cases}
\end{align}
where $Z$ is a zero-mean unit-variance Gaussian noise. In this case, the receiver $Y_2$ does not receive any information from the transmitter and only observes the additive noise of the receiver $Y_1$. The capacity region of this channel with degraded message sets can be derived by setting $a=1$, $b=0$ in (\ref{eq:FCgbc1})-(\ref{eq:FCgbc1_3}) as: %nd is given by
\begin{align}\label{eq:gex2}
     R_0 &\le C_{12}\\
    R_0+R_1 &\le \psi(P)+C_{21}.
\end{align}
\end{example}

This is the Gaussian counterpart of Example \ref{ex:DMC_1}. 
To achieve this capacity, first the receiver $Y_2$ sends a compressed version of the observed noise $Z$ to the receiver $Y_1$ through the digital link $C_{21}$. Next, the receiver $Y_1$ decodes both common and private messages based on the message from $Y_2$, then forwards the common message to $Y_2$ through the digital link $C_{12}$.

\begin{example}
\label{ex:GBC_3}
Consider the following Gaussian BC:
\begin{align}\label{eq:GBCex3}
\begin{cases}
    Y_1=Z\\
    Y_2=X+Z.
\end{cases}
\end{align}
The capacity region of this channel with degraded message sets is given by
\begin{align}\label{eq:gex3}
     R_0 &\le \psi(P)+C_{12}\\
    R_0+R_1 &\le C_{21}.
\end{align}
\end{example}

This is the Gaussian counterpart of Example \ref{ex:DMC_2}. 
Similar to Example \ref{ex:DMC_2}, the receiver $Y_1$ (which needs to detect both common and private messages) does not receive any information from the transmitter directly. The capacity achieving cooperation protocol is that the receiver $Y_2$ first sends a compressed version of its received signal to the receiver $Y_1$ through the digital link $C_{21}$. Next, the receiver $Y_1$ decodes both messages using the information received (and its own signal, which is in fact just the channel noise) then forwards the bin index of the decoded common message to the receiver $Y_2$ through the link $C_{12}$. Lastly, the receiver $Y_2$ decodes the common message using the received bin index and its own received signal. 

Note that in Example \ref{ex:GBC_3}, as the receiver $Y_1$ observes the channel noise only, one might think that a cooperative scheme in which $Y_1$ applies quantize-bin-and-forward first followed by $Y_2$ applying decode-and-forward should be used. However, such a scheme is not optimal in this particular BC where the message for $Y_2$ is a degraded version of the message for $Y_1$.
%where the message of $Y_2$ is a degraded version of the message for $Y_1$. 

Theorem \ref{th.FCgbc1} is a capacity result for the Gaussian BC with degraded message sets, and with bidirectional conferencing decoders and perfectly correlated noises. Next, we treat the case of a Gaussian BC with both common and private messages, but with one-sided cooperation. %and perfectly correlated noises. %we establish the capacity region of the channel 
%Next, we establish the capacity region of a Gaussian BC with both common and private messages, but with one-sided cooperation, and perfectly correlated noises. %we establish the capacity region of the channel 

\begin{theorem}\label{th.FCgbc2}
Consider the two-user Gaussian BC (\ref{eq:GBC}) with both common and private messages where $|a|\ge|b|$ and only the weaker receiver is connected to the stronger one by a conferencing link, i.e., $C_{12}=0$. Assume $P>0$. 
For the channel with perfectly correlated noises with $|\lambda|=1$, 
but $\lambda \neq \frac{b}{a}$, the capacity region is given by
the closure of the convex hull of all $(R_0,R_1,R_2)$ satisfying
\begin{align}\label{eq:FCgbc2}
    R_0+R_2 &\le \psi\left(\frac{(1-\beta)b^2P}{\beta b^2P+1}\right)\\
    R_0+R_1+R_2 &\le \psi\left(\beta a^2P\right)+\psi\left(\frac{(1-\beta)b^2P}{\beta b^2P+1}\right)+C_{21}
\end{align}
for some $\beta\in[0,1]$.
\end{theorem}

\begin{IEEEproof}
The proof %of Theorem \ref{th.FCgbc2} 
	is similar to that of Theorem \ref{th.FCgbc1}. % and therefore omitted.
\end{IEEEproof}

%The statement of Theorem \ref{th.FCgbc2} contains two cases. When $C_{12}=0$, i.e., the cooperation link is from the weaker receiver to the strong receiver, it turns out that superposition coding along with the compress-and-forward relaying strategy is capacity achieving when the noises are perfectly correlated.

%The case where $C_{21}=0$, i.e., the cooperation link is from the strong receiver to the weak receiver, is quite different. It turns out that we should use dirty-paper coding along with a relay strategy that involves decoding the private message then compressing a linear combination of both the decoded message and the received signal. Such a strategy is capacity achieving when the noise is correlated. 

%It is interesting to note that for the degraded Gaussian BC without cooperation, both superposition coding and dirty-paper coding can be used to achieve capacity. But with one-sided cooperation, one must choose the coding strategy judiciously.

%The above results are all for the Gaussian BC with perfectly correlated noises.
For the Gaussian BC in which the noises are not perfectly correlated, the inner and outer bounds of this paper can provide approximate capacity results. % presented in the following theorems.
The first such result is for the Gaussian BC with degraded message sets.

\begin{theorem}\label{th.AFCgbc1}
	Consider the two-user Gaussian BC (\ref{eq:GBC}) with degraded message sets (i.e., a common message for both receivers and a private message for the first receiver) and bidirectional conferencing decoders. Assume that $|a|\ge|b|$. % and $\lambda \neq \frac{b}{a}$.
For all channel parameters $a$, $b$, $C_{12}$, $C_{21}$, and $\lambda$ with $|\lambda|<1$, % and $\lambda \neq \frac{b}{a}$, 
the following achievable rate region is within half a bit from the capacity region:
\begin{align}\label{eq:AFCgbc1}
    R_0 &\le \psi\left(\frac{(1-\beta)b^2P}{\beta b^2P+1}\right)+C_{12}\\
    R_0+R_1 &\le \psi\left(a^2P\right)+\{C_{21}-1/2\}^+ \label{eq:AFCgbc1_2} \\
    R_0+R_1 &\le \psi\left(\left(\frac{a^2+b^2-2\lambda ab+(1-\lambda^2)a^2}{2(1-\lambda^2)}\right)P\right) \label{eq:AFCgbc1_3} \\
    R_0+R_1 &\le \psi\left(\beta a^2P\right) +\psi\left(\frac{(1-\beta)b^2P}{\beta b^2P+1}\right) \nonumber \\ 
& \qquad \qquad \qquad +\{C_{21}-1/2\}^+ +C_{12} \label{eq:AFCgbc1_4} \\
    R_0+R_1 &\le \psi\left(\beta\left(\frac{a^2+b^2-2\lambda ab+(1-\lambda^2)a^2}{2(1-\lambda^2)}\right)P\right) \nonumber \\
& \qquad \qquad \qquad +\psi\left(\frac{(1-\beta)b^2P}{\beta b^2P+1}\right)+C_{12} \label{eq:AFCgbc1_5}
\end{align}
for some $\beta\in[0,1]$.
\end{theorem}

\begin{IEEEproof}
The above achievable rate region is derived from $\mathcal{R}_i^{(1)}$ given in Proposition \ref{th.ibc1} by setting $X = U = W+\Bar{W}$ and $V = \hat{Y}_1=\varnothing$, and $\hat{Y}_2 = Y_2+\hat{Z}_2$,  where $W$ and $\Bar{W}$ are two independent Gaussian variables with zero means and variances $(1-\beta) P$ and $\beta P$, respectively, and $\hat{Z}_2$ is a Gaussian variable (independent of all other variables) with zero mean and variance $\hat{\sigma}^2=1-\lambda^2$. 

By a comparison, one can see that the right-hand sides of the constraints (\ref{eq:AFCgbc1}), (\ref{eq:AFCgbc1_2}), (\ref{eq:AFCgbc1_3}), (\ref{eq:AFCgbc1_4}) and (\ref{eq:AFCgbc1_5})
are within half a bit of (\ref{eq:ogbc_3}), (\ref{eq:ogbc}), (\ref{eq:ogbc_8}), (\ref{eq:ogbc_5}), and (\ref{eq:ogbc_6}), respectively. 
\end{IEEEproof}

Next, we present an approximate capacity result for the Gaussian BC with one-sided conferencing link. 

\begin{theorem}\label{th.AFCgbc2}
Consider the two-user Gaussian BC (\ref{eq:GBC}) with both common and private messages, where $|a|\ge|b|$ and only the weaker receiver has a conferencing link to the stronger receiver, i.e., $C_{12}=0$.  
For all channel parameters $a$, $b$, $C_{21}$, and $\lambda$ with $|\lambda|<1$, % and $\lambda \neq \frac{b}{a}$,
the following achievable rate region is within half a bit from the capacity region:
\begin{align}\label{eq:AFCgbc2}
    R_0+R_2 &\le \psi\left(\frac{(1-\beta)b^2P}{\beta b^2P+1}\right)\\
    R_0+R_1+R_2 &\le \psi\left(\left(\frac{a^2+b^2-2\lambda ab+(1-\lambda^2)a^2}{2(1-\lambda^2)}\right)P\right) \label{eq:AFCgbc2_2} \\
    R_0+R_1+R_2 &\le \psi\left(\beta a^2P\right) +\psi\left(\frac{(1-\beta)b^2P}{\beta b^2P+1}\right) \nonumber \\
& \qquad \qquad \qquad \qquad +\{C_{21}-1/2\}^+ \label{eq:AFCgbc2_3} \\
    R_0+R_1+R_2 &\le \psi\left(\beta\left(\frac{a^2+b^2-2\lambda ab+(1-\lambda^2)a^2}{2(1-\lambda^2)}\right)P\right) \nonumber \\
	& \qquad \qquad \qquad +\psi\left(\frac{(1-\beta)b^2P}{\beta b^2P+1}\right) \label{eq:AFCgbc2_4}
\end{align}
for some $\beta\in[0,1]$.
\end{theorem}

\begin{IEEEproof}
Similar to Theorem \ref{th.AFCgbc1}, the above achievable rate region is derived from $\mathcal{R}_i^{(1)}$ given in Proposition \ref{th.ibc1} by setting $X = U = W+\Bar{W}$ and $V=\hat{Y}_1=\varnothing$, and $\hat{Y}_2 = Y_2+\hat{Z}_2$,  where $W$ and $\Bar{W}$ are two independent Gaussian variables with zero means and variances $(1-\beta) P$ and $\beta P$, respectively, and $\hat{Z}_2$ is a Gaussian variable (independent of all other variables) with zero mean and variance $\hat{\sigma}^2 =1-\lambda^2$. 

By a simple comparison, one can see that the right-hand sides of the constraints (\ref{eq:AFCgbc2}), (\ref{eq:AFCgbc2_2}), (\ref{eq:AFCgbc2_3}) and (\ref{eq:AFCgbc2_4})
are within half a bit of (\ref{eq:ogbc_3}), (\ref{eq:ogbc_8}), (\ref{eq:ogbc_5}), and (\ref{eq:ogbc_6}), respectively. %$\blacksquare$
\end{IEEEproof}

Finally, we derive an approximate capacity result for the two-user Gaussian BC (\ref{eq:GBC}) with both common and private messages and with bidirectional cooperative receivers. First, we present an achievable rate region for the channel using only one-way conferencing with decode-and-forward. The other conferencing link is not used, so the resulting achievable rate region is a sub-region of $\mathcal{R}_i^{(1)}$ (\ref{eq:ibc1eq})-(\ref{eq:ibc1eq_end}). %(\ref{eq:ibc1})-(\ref{eq:ibc1_end}). 
It turns out that this region is already approximately optimal when the noise correlation is small.

\begin{proposition}\label{DFgbc}
Let $\mathcal{R}_i^{DF-G}$ denote the set of all rate triples $(R_0,R_1,R_2)$ such that
\begin{align}\label{eq:DFgbc}
    R_0+R_2 &\le \psi\left(\frac{(1-\beta)b^2P}{\beta b^2P+1}\right)+C_{12}\\
    R_0+R_1+R_2 &\le \psi\left(a^2P\right) \label{eq:DFgbc_2} \\
    R_0+R_1+R_2 &\le \psi\left(\beta a^2P\right)+\psi\left(\frac{(1-\beta)b^2P}{\beta b^2P+1}\right)+C_{12} \label{eq:DFgbc_3}
\end{align}
for some $\beta\in[0,1]$. The set $\mathcal{R}_i^{DF-G}$ constitutes an inner bound on the capacity region of the Gaussian BC (\ref{eq:GBC}) with conferencing decoders.
\end{proposition}

\begin{IEEEproof}
%\textit{Proof of Corollary \ref{DFgbc}}: 
The bound $\mathcal{R}_i^{DF-G}$ is derived from $\mathcal{R}_i^{(1)}$ given in Theorem \ref{alpha_1*} by setting $X = U = W+\Bar{W}$ and $V = \hat{Y}_1 = \hat{Y}_2 = \varnothing$, where $W$ and $\Bar{W}$ are two independent Gaussian variables with zero means and variances $\beta P$ and $(1-\beta)P$, respectively. Note that this inner bound is in fact derived for the channel with one-sided cooperation (i.e., it does not make use of the conferencing link $C_{21}$) using the decode-and-forward strategy alone. Nevertheless, it is a valid inner bound for the channel with bidirectional cooperation. %$\blacksquare$
\end{IEEEproof}

\begin{theorem}\label{th.app.Gg}
Consider the two-user Gaussian BC (\ref{eq:GBC}) with both common and private messages and with bidirectional conferencing decoders. % as shown in Fig.~\ref{fig:BC_CD}. 
Assume that $|a|\ge|b|$. For all channel parameters $a$, $b$, $C_{12}$, $C_{21}$, and $\lambda$ with $|\lambda|<1$, % and $\lambda \neq \frac{b}{a}$,
the inner bound $\mathcal{R}_i^{DF-G}$ is within $\frac{1}{2}\log (\frac{2}{1-\left|\lambda\right|})$ bits of the capacity region.
\end{theorem}

\begin{IEEEproof}
%\textit{Proof of Theorem \ref{th.app.Gg}:} 
The constraint (\ref{eq:DFgbc}) is identical to (\ref{eq:ogbc_3}). Moreover, by simple algebraic computation, one can show that the right-hand sides of the constraints (\ref{eq:DFgbc_2}) and (\ref{eq:DFgbc_3}) are within $\frac{1}{2}\log (\frac{2}{1-\left|\lambda\right|})$ bits of (\ref{eq:ogbc_8}) and (\ref{eq:ogbc_6}), respectively. %$\blacksquare$
\end{IEEEproof}

\begin{remark}
For the Gaussian BC with $\lambda ab \ge 0$, a better approximate capacity gap bound of $\frac{1}{2}\log (\frac{2}{1-\lambda^2})$ bits is possible for the region $\mathcal{R}_i^{DF-G}$ given in Proposition \ref{DFgbc}. 
%result is possible. Specifically, if $\lambda ab \ge 0$, one can show that the inner bound $\mathcal{R}_i^{DF-G}$ given in Corollary \ref{DFgbc} is within $\frac{1}{2}\log (\frac{2}{1-\lambda^2})$ bits of the capacity region.
\end{remark}

Theorem \ref{th.app.Gg} states that if we do not use the quantize-bin-and-forward part of the conferencing protocol and rely solely on decode-and-forward, the gap to capacity would depend on the noise correlation. This is because decode-and-forward cannot exploit the noise correlation. %so while it achieves within constant gap to the capacity region when the noises are uncorrelated, it cannot do so when the noises are highly correlated. 
Nevertheless, there is one special case for which decode-and-forward is optimal.

\begin{remark} \label{degraded2}
For the Gaussian BC (\ref{eq:GBC}) with $|a| \ge |b|$, when $\lambda=\frac{b}{a}$, 
one can verify that the BC is physically degraded (see Appendix 
\ref{app:outer_bound_gbc}) and the decode-and-forward achievable region $\mathcal{R}_i^{DF-G}$ of Proposition \ref{DFgbc} coincides with the outer bound $\mathcal{R}_o^G$ in Theorem \ref{th.ogbc}; 
(see (\ref{eq:condition_lambda_ab})).
Thus,  $\mathcal{R}_i^{DF-G}$ is the capacity region.
\end{remark}

As a concluding remark for this section, the novel structure of the outer bound given in Theorem \ref{th.ogbc}, in particular, the constraint (\ref{eq:ogbc_6}), is crucial for deriving the exact capacity results in Theorems \ref{th.FCgbc1} and \ref{th.FCgbc2}, and also the approximate capacity results in Theorems \ref{th.AFCgbc1}, \ref{th.AFCgbc2}, and \ref{th.app.Gg}.

\section{Conclusions}
\label{sec:conclude}

Receiver cooperation can significantly improve the capacity of a broadcast
channel.  In this paper, we establish novel outer bounds, which are tighter
than the cut-set bound, for the BC with bidirectional conferencing decoders
using the Csisz\'{a}r-K\"{o}rner identity.  Together with the achievability
results based on Marton's coding and judicious combination of
quantize-bin-and-forward and decode-and-forward, we derive capacity results for
specific classes of semi-deterministic BCs with degraded message sets and more
capable semi-deterministic BCs with one-sided conferencing.  Furthermore, for
the Gaussian BC, we illustrate the importance of noise correlation when the
receivers can cooperate using the conferencing links.  For a Gaussian BC with
perfectly negatively correlated noises, an infinitesimal amount of power is
already sufficient for achieving a rate equal to the capacity of the
conferencing links.

\appendices

\section{Proof of Theorem \ref{th.obc}}
\label{app:outer_bound}

We only derive the constraints that include the auxiliary variable $U$, i.e., (\ref{eq:outer_R1U}), (\ref{eq:outer_R2U_1}), (\ref{eq:outer_R2U_2}), (\ref{eq:outer_RsumU_1}), and (\ref{eq:outer_RsumU_2}). The constraints that include $V$ can be derived symmetrically. The last constraint (\ref{eq:outer_cs}) is due to the cut-set bound.

Consider a length-$n$ code with common and private message rates $(R_0, R_1, R_0)$ and with vanishing average probability of decoding error. 
Define the auxiliary variables as follows:
\begin{equation}
U_t=(M_0,M_1,Y_2^{t-1},Y_{1,t+1}^n), \qquad t=1,...,n.
\label{eq:Ut}
\end{equation}
By Fano's inequality, we have
\begin{align}
n&(R_0+R_1) \nonumber \\
&\le I(M_0,M_1;Y_1^n,J_{21}^{\Gamma})+n\epsilon_n^1
\\&=I(M_0,M_1;Y_1^n)+I(M_0,M_1;J_{21}^{\Gamma}|Y_1^n)+n\epsilon_n^1\\
&\le \sum_{t=1}^n I(M_0,M_1;Y_{1,t}|Y_{1,t+1}^n)+H(J_{21}^{\Gamma})+n\epsilon_n^1\\
&\le \sum_{t=1}^n I(M_0,M_1,Y_{1,t+1}^n,Y_2^{t-1};Y_{1,t})+H(J_{21}^{\Gamma})+n\epsilon_n^1 \label{eq:fano_a} \\
&\le \sum_{t=1}^nI(U_t;Y_{1,t})+nC_{21}+n\epsilon_n^1
\end{align}
where inequality $(\ref{eq:fano_a})$ holds because conditioning does not increase the entropy. This corresponds to (\ref{eq:outer_R1U}). 

Next, we derive the constraints on $R_2$. By Fano's inequality, 
\begin{align}
    nR_2&\le I(M_2;Y_2^n,J_{12}^{\Gamma})+n\epsilon_n^2
    \\&\le I(M_2;Y_1^n,Y_2^n,J_{12}^{\Gamma})+n\epsilon_n^2\\
    &= I(M_2;Y_1^n,Y_2^n)+n\epsilon_n^2 \label{eq:fano_b} \\
    &\le I(M_2;Y_1^n,Y_2^n|M_0,M_1)+n\epsilon_n^2
\label{eq:converse_R2}
\end{align}
where $(\ref{eq:fano_b})$ holds because $J_{12}^{\Gamma}$ is given by a deterministic function of $(Y_1^n,Y_2^n)$. 

Now consider the following:
\begin{align}
%\begin{split}
    I&(M_2;Y_1^n,Y_2^n|M_0,M_1) \nonumber \\
    &=I(M_2;Y_2^n|M_0,M_1)+I(M_2;Y_1^n|Y_2^n,M_0,M_1)\\
    &\le I(M_0,M_1,M_2;Y_2^n)+I(M_2;Y_1^n|Y_2^n,M_0,M_1) \\
    &\le \sum_{t=1}^{n}I(X_t;Y_{2,t})  \nonumber \\ & \qquad  +\sum_{t=1}^{n}I(X_t;Y_{1,t}|Y_{2,t},M_0,M_1,Y_2^{t-1},Y_{2,t+1}^n,Y_{1,t+1}^n) \label{eq:converse_M2_a} \\
    &=\sum_{t=1}^{n}I(X_t;Y_{2,t})+\sum_{t=1}^{n}I(X_t;Y_{1,t}|Y_{2,t},U_t,Y_{2,t+1}^n)\\
    &\le \sum_{t=1}^{n}I(X_t;Y_{2,t})+\sum_{t=1}^{n}I(X_t;Y_{1,t}|Y_{2,t},U_t),
\label{eq:converse_M2_b}
%\end{split}
\end{align}
where $(\ref{eq:converse_M2_a})$ holds because $X_t$ is a deterministic function of $(M_0,M_1,M_2)$ and the channel is memoryless, and $(\ref{eq:converse_M2_b})$ holds because conditioning does not increase entropy and also given the input signal $X_t$, the outputs $Y_{1,t},Y_{2,t}$ are independent of other variables. 

In the same way, one can also derive
\begin{align}
    I&(M_2;Y_1^n,Y_2^n|M_0,M_1) \nonumber \\
    &=I(M_2;Y_1^n|M_0,M_1)+I(M_2;Y_2^n|Y_1^n,M_0,M_1)\\
    &\le \sum_{t=1}^{n}I(X_t;Y_{1,t})+\sum_{t=1}^{n}I(X_t;Y_{2,t}|Y_{1,t},U_t).
\label{eq:converse_M2}
\end{align}
By substituting (\ref{eq:converse_M2_b}) and (\ref{eq:converse_M2}) into (\ref{eq:converse_R2}), we obtain the desired bounds on $R_2$ (\ref{eq:outer_R2U_1}) and  
(\ref{eq:outer_R2U_2}).

Next, we derive the constraints on the sum rate, starting from Fano's inequality:
\begin{align}
    n&(R_0+R_1+R_2) \nonumber \\
    &\le I(M_0,M_1;Y_1^n,J_{21}^{\Gamma})+I(M_2;Y_2^n,J_{12}^{\Gamma})+n(\epsilon_n^1+\epsilon_n^2)\\
    &\le I(M_0,M_1;Y_1^n)+I(M_2;Y_2^n)+I(M_0,M_1;J_{21}^{\Gamma}|Y_1^n) \nonumber \\
    &\qquad +I(M_2;J_{12}^{\Gamma}|Y_2^n)+n(\epsilon_n^1+\epsilon_n^2)\\
    &\le I(M_0,M_1;Y_1^n)+I(M_2;Y_2^n|M_0,M_1) \nonumber \\
    &\qquad +H(J_{21}^{\Gamma})
    +H(J_{12}^{\Gamma})+n(\epsilon_n^1+\epsilon_n^2) \label{eq:129}\\
    &\le I(M_0,M_1;Y_1^n)+I(M_2;Y_2^n|M_0,M_1) \nonumber \\
    &\qquad +n(C_{12}+C_{21})+n(\epsilon_n^1+\epsilon_n^2),
\label{eq:converse_R0R1R2}
\end{align}
where in \eqref{eq:129} we use the fact that $M_0$, $M_1$ and $M_2$ are independent to
conclude that $I(M_2;Y_2^n) \le I(M_2;Y_2^n|M_0,M_1)$.

Now consider the following:
\begin{align}
    I&(M_0,M_1;Y_1^n)+I(M_2;Y_2^n|M_0,M_1) \nonumber \\
	&=\sum_{t=1}^{n}I(M_0,M_1;Y_{1,t}|Y_{1,t+1}^n) \nonumber \\
    &\qquad +\sum_{t=1}^{n}I(M_2;Y_{2,t}|M_0,M_1,Y_2^{t-1})\\
    &\le \sum_{t=1}^{n}I(Y_2^{t-1},M_0,M_1;Y_{1,t}|Y_{1,t+1}^n) \nonumber \\
    &\qquad -\sum_{t=1}^{n}I(Y_2^{t-1};Y_{1,t}|M_0,M_1,Y_{1,t+1}^n) \nonumber \\
    &\qquad +\sum_{t=1}^{n}I(Y_{1,t+1}^n,M_2;Y_{2,t}|M_0,M_1,Y_2^{t-1})\\
    &= \sum_{t=1}^{n}I(Y_2^{t-1},M_0,M_1;Y_{1,t}|Y_{1,t+1}^n) \nonumber \\
    &\qquad -\sum_{t=1}^{n}I(Y_2^{t-1};Y_{1,t}|M_0,M_1,Y_{1,t+1}^n) \nonumber \\
    &\qquad +\sum_{t=1}^{n}I(Y_{1,t+1}^n;Y_{2,t}|M_0,M_1,Y_2^{t-1}) \nonumber \\
    &\qquad +\sum_{t=1}^{n}I(M_2;Y_{2,t}|M_0,M_1,Y_2^{t-1},Y_{1,t+1}^n)\\
    &\le \sum_{t=1}^{n}I(Y_2^{t-1},M_0,M_1;Y_{1,t}|Y_{1,t+1}^n) \nonumber \\
    &\qquad +\sum_{t=1}^{n}I(M_2;Y_{2,t}|M_0,M_1,Y_2^{t-1},Y_{1,t+1}^n) 
\label{eq:converse_M0M1_a} \\
    &\le \sum_{t=1}^{n}I(Y_2^{t-1},Y_{1,t+1}^n,M_0,M_1;Y_{1,t}) \nonumber \\
    &\qquad +\sum_{t=1}^{n}I(M_2;Y_{2,t}|M_0,M_1,Y_2^{t-1},Y_{1,t+1}^n)\\
    &\le \sum_{t=1}^{n}I(U_t;Y_{1,t}) + \sum_{t=1}^{n} H(Y_{2,t}|U_t) \nonumber \\
    &\qquad - \sum_{t=1}^{n} H(Y_{2,t}|M_0,M_1,M_2,Y_2^{t-1},Y_{1,t+1}^n)
\label{eq:converse_M0M1_b} \\
    &\le \sum_{t=1}^{n}I(U_t;Y_{1,t}) + \sum_{t=1}^{n} H(Y_{2,t}|U_t) \nonumber \\
    &\qquad - \sum_{t=1}^{n} H(Y_{2,t}|X_t,M_0,M_1,M_2,Y_2^{t-1},Y_{1,t+1}^n)
\label{eq:converse_M0M1_c} \\
    &\le \sum_{t=1}^{n}I(U_t;Y_{1,t}) + \sum_{t=1}^{n} H(Y_{2,t}|U_t) \nonumber \\
    &\qquad - \sum_{t=1}^{n} H(Y_{2,t}|X_t,M_0,M_1,Y_2^{t-1},Y_{1,t+1}^n)
\label{eq:converse_M0M1_d} \\
    &\le \sum_{t=1}^{n}I(U_t;Y_{1,t})+\sum_{t=1}^{n}I(X_t;Y_{2,t}|U_t)
\label{eq:converse_M0M1}
\end{align}
where (\ref{eq:converse_M0M1_a}) is due to the Csisz\'{a}r-K\"{o}rner identity
\cite{marton} \cite[Lemma 7]{CsisKo}, which %. The Csisz\'{a}r-K\"{o}rner identity
 states that given two arbitrary random vectors $A_1^n$ and $B_1^n$ and an arbitrary random variable $C$, the following is true: 
\begin{align}
    \centering
    \sum_{t=1}^n I(A_{t+1}^n;B_t|B_1^{t-1},C)=\sum_{t=1}^n I(B_1^{t-1};A_t|A_{t+1}^
n,C),
\end{align}
and \eqref{eq:converse_M0M1_b} is due to the definition of $U_t$, 
\eqref{eq:converse_M0M1_c} is due to that conditioning does not increase entropy,  
\eqref{eq:converse_M0M1_d} is due to the Markov chain 
$(M_2,U_t) \rightarrow X_t \rightarrow Y_{2,t}$. 
By substituting (\ref{eq:converse_M0M1}) in (\ref{eq:converse_R0R1R2}), we obtain the bound (\ref{eq:outer_RsumU_1}) on the sum rate. 

Finally, we can write
\begin{align}
%\begin{split}
    n&(R_0+R_1+R_2) \nonumber \\
    &\le I(M_0,M_1;Y_1^n,J_{21}^{\Gamma})+I(M_2;Y_2^n,J_{12}^{\Gamma})+n(\epsilon_n^1+\epsilon_n^2)\\
    &\le I(M_0,M_1;Y_1^n) +I(M_0,M_1;J_{21}^{\Gamma}|Y_1^n) \nonumber \\
    &\qquad + I(M_2;Y_1^n,Y_2^n,J_{12}^{\Gamma}) +n(\epsilon_n^1+\epsilon_n^2)\\
    &\le I(M_0,M_1;Y_1^n)+ H(J_{21}^{\Gamma}) + I(M_2;Y_1^n,Y_2^n) \nonumber \\
	& \qquad +n(\epsilon_n^1+\epsilon_n^2)\\
    &\le I(M_0,M_1;Y_1^n) +nC_{21} + I(M_2;Y_1^n,Y_2^n|M_0,M_1) \nonumber \\
	\label{eq:converse_R0R1R2_aa} 
    &\qquad +n(\epsilon_n^1+\epsilon_n^2)\\
    &\le I(M_0,M_1;Y_1^n) +nC_{21} +I(M_2;Y_2^n|M_0,M_1)  \nonumber \\
    &\qquad+I(M_2;Y_1^n|Y_2^n,M_0,M_1) %\nonumber \\ &\qquad 
	+n(\epsilon_n^1+\epsilon_n^2)\\
    &\le \sum_{t=1}^{n}I(U_t;Y_{1,t})+\sum_{t=1}^{n}I(X_t;Y_{2,t}|U_t)  
	\nonumber \\
	\label{eq:converse_R0R1R2_a} 
    &\qquad+\sum_{t=1}^{n}I(X_t;Y_{1,t}|Y_{2,t},U_t)+nC_{21}+n(\epsilon_n^1+\epsilon_n^2)\\
    &= \sum_{t=1}^{n}I(U_t;Y_{1,t})+\sum_{t=1}^{n}I(X_t;Y_{1,t},Y_{2,t}|U_t)  \nonumber \\
    &\qquad+nC_{21}+n(\epsilon_n^1+\epsilon_n^2)\\
    &= \sum_{t=1}^{n}I(U_t;Y_{1,t})+\sum_{t=1}^{n}I(X_t;Y_{1,t}|U_t) \nonumber \\
    &\qquad+\sum_{t=1}^{n}I(X_t;Y_{2,t}|Y_{1,t},U_t)+nC_{21}+n(\epsilon_n^1+\epsilon_n^2)\\
    &= \sum_{t=1}^{n}I(U_t,X_t;Y_{1,t})+\sum_{t=1}^{n}I(X_t;Y_{2,t}|Y_{1,t},U_t)  \nonumber \\
    &\qquad +nC_{21}+n(\epsilon_n^1+\epsilon_n^2) \\
    &= \sum_{t=1}^{n}I(X_t;Y_{1,t})+\sum_{t=1}^{n}I(X_t;Y_{2,t}|Y_{1,t},U_t)  \nonumber \\
	\label{eq:converse_R0R1R2_b} 
    &\qquad +nC_{21}+n(\epsilon_n^1+\epsilon_n^2)
%\end{split}
\end{align}
where inequality \eqref{eq:converse_R0R1R2_aa} is due to that $M_0$, $M_1$ and $M_2$ are
independent so $ I(M_2;Y_1^n,Y_2^n|M_0,M_1) \ge I(M_2;Y_1^n,Y_2^n)$,
inequality $(\ref{eq:converse_R0R1R2_a})$ 
is derived by following the same line of argument as in (\ref{eq:converse_M0M1})
using the Csisz\'{a}r-K\"{o}rner identity, 
and \eqref{eq:converse_R0R1R2_b} is due to the Markov chain $U_t \rightarrow X_t \rightarrow Y_{1,t}$.
This corresponds to the desired constraint (\ref{eq:outer_RsumU_2}). 

%Lastly, (\ref{eq:outer_cs}) is just the cut-set bound. 
By applying a standard time-sharing argument, the proof is thus complete. %$\blacksquare$
The cardinality bounds on $U$ and $V$ are due to Carathéodory theorem \cite{NIT}.

\section{Proof of Proposition \ref{th.ibc1}}
\label{app:achievability_1}

We prove the achievability of $\mathcal{R}_i^{(1)}$ using Marton's coding for the BC together with quantize-bin-and-forward at $Y_2$, then a combination of decode-and-forward and quantize-bin-and-forward at $Y_1$. Fix a joint distribution $$P(u,v,w,x)P(\hat{y}_1|u,w,y_1)P(\hat{y}_2|y_2)$$ and let $\alpha_1 \in [0,1]$. We construct random codebooks consisting of length-$n$ codewords to transmit the common message $M_0$, and the private messages $M_1$ and $M_2$, which are independent and uniformly distributed over the sets $[1:2^{nR_0}]$, $[1:2^{nR_1}]$, and $[1:2^{nR_2}]$, respectively.

First of all, each of the private messages $M_1$ and $M_2$ is split into two parts as follows:
\begin{align}
    M_1&\longmapsto (M_{10},M_{11}),
\end{align}
where $M_{10}\in [1:2^{nR_{10}}], M_{11}\in [1:2^{nR_{11}}]$, and 
\begin{align}
    M_2&\longmapsto (M_{20},M_{22}),
\end{align}
where $M_{20}\in [1:2^{nR_{20}}], M_{22}\in [1:2^{nR_{22}}]$.  
Therefore, 
\begin{align}
        R_1=R_{10}+R_{11}, \\
        R_2=R_{20}+R_{22}.
\end{align}
The sub-messages $M_{10}$ and $M_{20}$ are decoded at both receivers, i.e., the triple $(M_0,M_{10},M_{20})$ is considered as the common message with rate 
\begin{align}
        R_{00}=R_0+R_{10}+R_{20}.
\end{align}
Moreover, let $B_1, B_2, \hat{R}_1, \hat{R}_2$ be non-negative real numbers
denoting the rates of the binning steps and the compression steps in the proof. 
Also, define
\begin{align}\label{eq:Cind}
        \hat{C}_{12}^d&=\min \{\Bar{\alpha}_1C_{12},R_{00}\}\\
        \hat{C}_{12}^q&=\min \{{\alpha}_1C_{12} \hat{R}_1\}\label{eq:ApBC12qhatmin}\\
        \hat{C}_{21}&=\min \{C_{21},\hat{R}_2\} \label{eq:ApBC21hatmin}
\end{align}
as the portion of the conferencing link rates used in the two directions and for decode-and-forward and quantize-bin-and-forward, respectively.

\subsection{Encoding at the Transmitter}

\begin{enumerate}
\item Generate at random $2^{nR_{00}}$ independent codewords $W^n$ according to $P(w^n)=\prod_{t=1}^n P(w_t)$. Label these codewords $W^n(m_{00})$ where $m_{00}\in [1:2^{nR_{00}}]$.

\item For each codeword $W^n(m_{00})$, generate at random $2^{n(R_{11}+B_1)}$ independent codewords $U^n$ according to $\prod_{t=1}^nP(u_t|w_t(m_{00}))$. Label these codewords $U^n(m_{00},m_{11},b_1)$ where $m_{11}\in[1:2^{nR_{11}}]$ and $b_1\in[1:2^{nB_1}]$.

\item For each codeword $W^n(m_{00})$, generate at random $2^{n(R_{22}+B_2)}$ independent codewords $V^n$ according to $\prod_{t=1}^nP(v_t|w_t(m_{00}))$. Label these codewords $V^n(m_{00},m_{22},b_2)$ where $m_{22}\in[1:2^{nR_{22}}]$ and $b_2\in[1:2^{nB_2}]$.

\item Given a triple $(m_{00},m_{11},m_{22})$ to be transmitted, find a pair of indices $(b_1^\mathcal{T},b_2^\mathcal{T})$ that satisfies 
\begin{multline}\label{eq:typenc}
(W^n(m_{00}),U^n(m_{00},m_{11},b_1^\mathcal{T}),V^n(m_{00},m_{22},b_2^\mathcal{T})) \\
\qquad \in\mathcal{T}_\epsilon^n(P_{WUV}).
\end{multline}
If there are multiple pairs of $(b_1,b_2)$ that satisfy (\ref{eq:typenc}), let $(b_1^\mathcal{T},b_2^\mathcal{T})$ be the one with minimum $b_1$. 
If multiple such pairs have the same minimum $b_1$, choose the one with the minimum $b_2$. 
If no pair $(b_1,b_2)$ satisfies (\ref{eq:typenc}), declare an encoding error. Based on the mutual covering lemma \cite{NIT}, one can easily see that for sufficiently large $n$ and small $\epsilon$, the probability of encoding error tends to zero provided that
\begin{align}\label{eq:ApBbiningbound}
         B_1+B_2\ge I(U;V|W).
\end{align}
The transmitter then generates a codeword $X^n$ according to $$\prod_{t=1}^nP(x_t|w_t(m_{00}),u_t(m_{00},m_{11},b_1^\mathcal{T}),v_t(m_{00},m_{22},b_2^\mathcal{T}))$$ and sends it over the channel.
\end{enumerate}

\subsection{Encoding and Decoding at the Receivers}

\begin{enumerate}
\item Generate at random $2^{n\hat{R}_2}$ independent codewords $\hat{Y}_2^n$ according to the distribution $\prod_{t=1}^nP(\hat{y}_{2,t})$. Label these codewords $\hat{Y}_2^n(\hat{m}_2)$ where $\hat{m}_2\in [1:2^{n\hat{R}_2}]$. 

\item Partition the set $[1:2^{n\hat{R}_2}]$ into $2^{n\hat{C}_{21}}$ bins, each containing $2^{n(\hat{R}_2-\hat{C}_{21})}$ elements. For a given $\hat{m}_2\in [1:2^{n\hat{R}_2}]$, let $\mathfrak{C}_{21}(\hat{m}_2)$ denote the index of the bin that $\hat{m}_2$ belongs to. 

\item Given the output sequence $Y_2^n$ at the second receiver, let $\hat{m}_2^\mathcal{T}$ be the smallest integer in $[1:2^{n\hat{R}_2}]$ that satisfies 
\begin{align}\label{eq:Y2Y2hattyp}
        (Y_2^n,\hat{Y}_2^n(\hat{m}_2^\mathcal{T}))\in \mathcal{T}_\epsilon^n(P_{Y_2\hat{Y}_2}).
\end{align}
If there is no integer in $[1:2^{n\hat{R}_2}]$ that satisfies (\ref{eq:Y2Y2hattyp}), then declare an error. The probability of this error tends to zero provided that
\begin{align}\label{eq:R2hat}
        \hat{R}_2\ge I(Y_2;\hat{Y}_2).
\end{align}
The receiver $Y_2$ then sends $\kappa= \mathfrak{C}_{21}(\hat{m}_2^\mathcal{T})$ to the receiver $Y_1$ through the digital link.

\item Given its output sequence $Y_1^n$ and $\kappa \in [1:2^{n\hat{C}_{21}}]$, the first receiver seeks a unique triple $(m_{00}^*,m_{11}^*,b_1^*)$ for which there exists some $\hat{Y}_2^n$ that satisfies 
\begin{multline}\label{eq:Y1dec}
    (W^n(m_{00}^*),U^n(m_{00}^*,m_{11}^*,b_1^*),Y_1^n,\hat{Y}_2^n)\\ \in \mathcal{T}_\epsilon^n(P_{WUY_1\hat{Y}_2}) 
\end{multline}
and
\begin{align}\label{eq:Y1dec_2}
        \mathfrak{C}_{21}(\hat{Y}_2^n)=\kappa.
\end{align}
If there is such a unique triple, the first receiver decodes its intended messages as $m_{00}^*$ and $m_{11}^*$. Otherwise, it declares an error. The probability that the actual transmitted messages do not satisfy (\ref{eq:Y1dec})-(\ref{eq:Y1dec_2}) is negligible for sufficiently large $n$. Consider the error event that some messages other than the actual transmitted ones together with some $\hat{Y}_2^n$ would satisfy (\ref{eq:Y1dec})-(\ref{eq:Y1dec_2}). This event comprises two cases. First, this $\hat{Y}_2^n$ is not the codeword actually selected from the quantization codebook, i.e., $\hat{Y}_2^n(\hat{m}_2^\mathcal{T})$. Second, this $\hat{Y}_2^n$ is actually $\hat{Y}_2^n(\hat{m}_2^\mathcal{T})$. Using the multivariate packing lemma \cite{NIT}, one can show that the probability of the first case tends to zero provided that
\begin{align}\label{eq:Y1bounds1}
\hat{R}_2-\hat{C}_{21}+&R_{11}+B_1 \le \nonumber \\
	&I(U;Y_1|W)+I(W,U,Y_1;\hat{Y}_2)\\
\hat{R}_2-\hat{C}_{21}+&R_{00}+R_{11}+B_1 \le \nonumber \\ 
	&I(W,U;Y_1)+I(W,U,Y_1;\hat{Y}_2)\label{eq:Y1bounds1_2}
\end{align}
and the probability of the second case tends to zero provided that
\begin{align}\label{eq:Y1bounds2}
        R_{11}+B_1\le I(U;Y_1,\hat{Y}_2|W)\\
        R_{00}+R_{11}+B_1\le I(W,U;Y_1,\hat{Y}_2).\label{eq:Y1bounds2_2}
\end{align}

\item Generate at random $2^{n\hat{R}_1}$ independent codewords $\hat{Y}_1^n$ according to the distribution $\prod_{t=1}^nP(\hat{y}_{1,t})$. Label these codewords $\hat{Y}_1^n(\hat{m}_1)$ where $\hat{m}_1\in [1:2^{n\hat{R}_1}]$. 

\item Partition the set $[1:2^{n\hat{R}_1}]$ into $2^{n\hat{C}_{12}^q}$ bins, each containing $2^{n(\hat{R}_1-\hat{C}_{12}^q)}$ elements. For a given $\hat{m}_1\in [1:2^{n\hat{R}_1}]$, let $\mathfrak{C}_{12}^q(\hat{m}_1)$ denote the index of the bin that $\hat{m}_1$ belongs to. 

\item Given the output sequence $Y_1^n$ at the first receiver, let $\hat{m}_1^\mathcal{T}$ be the smallest integer in $[1:2^{n\hat{R}_1}]$ that satisfies 
\begin{multline}\label{eq:Y1hatsel}
(W^n(m_{00}^*),U^n(m_{00}^*,m_{11}^*,b_1^*),Y_1^n,\hat{Y}_1^n(\hat{m}_1^\mathcal{T})) \\
	\in \mathcal{T}_\epsilon^n(P_{WUY_1\hat{Y}_1}). 
\end{multline}
Note that $(m_{00}^*,m_{11}^*,b_1^*)$ is the unique triple decoded in Step 4 according to (\ref{eq:Y1dec})-(\ref{eq:Y1dec_2}). If there is no integer in $[1:2^{n\hat{R}_1}]$ that satisfies (\ref{eq:Y1hatsel}), then declare an error. The probability of this error event tends to zero provided that
\begin{align}\label{eq:R1hat}
        \hat{R}_1\ge I(W,U,Y_1;\hat{Y}_1).
\end{align}

\item Partition the set $[1:2^{nR_{00}}]$ into $2^{n\hat{C}_{12}^d}$ bins, each containing $2^{n(R_{00}-\hat{C}_{12}^d)}$ elements. For a given $m_{00}\in [1:2^{nR_{00}}]$, let $\mathfrak{C}_{12}^d(m_{00})$ denote the index of the bin that $m_{00}$ belongs to.

\item The receiver $Y_1$ then sends the pair $(\theta_1,\theta_2)=(\mathfrak{C}_{12}^d(m_{00}^*),\mathfrak{C}_{12}^q(\hat{m}_1^\mathcal{T}))$ to the receiver $Y_2$ through the digital link. (Note that $m_{00}^*$ is the decoded message in Step 4 at the receiver $Y_1$.)

\item Given $(\theta_1,\theta_2)\in [1:2^{n\hat{C}_{12}^d}]\times [1:2^{\hat{C}_{12}^q}]$ and its output sequence $Y_2^n$, the second receiver seeks a unique triple $(m_{00}^\bullet,m_{22}^\bullet,b_2^\bullet)$ for which there exists some $\hat{Y}_1^n$ that satisfies 
\begin{multline}\label{eq:Y2dec}
    (W^n(m_{00}^\bullet),V^n(m_{00}^\bullet,m_{22}^\bullet,b_2^\bullet),\hat{Y}_1^n,Y_2^n) \\
\in \mathcal{T}_\epsilon^n(P_{WV\hat{Y}_1Y_2}),
\end{multline}
and
\begin{align}\label{eq:Y2dec_2}
 \mathfrak{C}_{12}^d(m_{00}^\bullet) &= \theta_1,  \\
 \mathfrak{C}_{12}^q(\hat{Y}_1^n) &=\theta_2.\label{eq:ApBC12qY1hattheta2} 
\end{align}
If there is such a unique triple, the second receiver decodes the its intended messages as $m_{00}^\bullet$ and $m_{22}^\bullet$. Otherwise, it declares an error. The probability that the actual transmitted messages do not satisfy (\ref{eq:Y2dec})-(\ref{eq:ApBC12qY1hattheta2}) is negligible for sufficiently large $n$. Consider the error event that some messages other than the actual transmitted ones together with some $\hat{Y}_1^n$ would satisfy (\ref{eq:Y2dec})-(\ref{eq:ApBC12qY1hattheta2}). This event comprises two cases. First, this $\hat{Y}_1^n$ is not the codeword actually selected from the quantization codebook, i.e., $\hat{Y}_1^n(\hat{m}_1^\mathcal{T})$. Second, this $\hat{Y}_1^n$ is actually $\hat{Y}_1^n(\hat{m}_1^\mathcal{T})$. Using the multivariate packing lemma \cite{NIT}, one can show that the probability of the first case tends to zero provided that
\begin{align}\label{eq:Y2bnds1}
\hat{R}_1-\hat{C}_{12}^q+&R_{22}+B_2\le \nonumber \\
	&I(V;Y_2|W)+I(W,V,Y_2;\hat{Y}_1)\\
\hat{R}_1-\hat{C}_{12}^q+&R_{00}-\hat{C}_{12}^d+R_{22}+B_2\le \nonumber \\
	&I(W,V;Y_2)+I(W,V,Y_2;\hat{Y}_1),\label{eq:Y2bnds1_2}
\end{align}
and the probability of the second case tends to zero provided that
\begin{align}\label{eq:Y2bnds2}
        R_{22}+B_2\le I(V;\hat{Y}_1,Y_2|W)\\
        R_{00}-\hat{C}_{12}^d+R_{22}+B_2\le I(W,V;\hat{Y}_1,Y_2).\label{eq:Y2bnds2_2}
\end{align}

\end{enumerate}

\subsection{Final Analysis}

Using the Fourier-Motzkin algorithm to eliminate $\hat{R}_2$ and $\hat{C}_{21}$ and by algebraic computation, one can see that (\ref{eq:ApBC21hatmin}), (\ref{eq:R2hat}), (\ref{eq:Y1bounds1})-(\ref{eq:Y1bounds1_2}), and (\ref{eq:Y1bounds2})-(\ref{eq:Y1bounds2_2}) can be equivalently rewritten in the following form:
%the following:
\begin{align}\label{eq:Y1BndFin}
       R_{11}+B_1&\le I(U;Y_1|W)+I(W,U,Y_1;\hat{Y}_2) \nonumber \\ 
	& \qquad \qquad +C_{21}-I(Y_2;\hat{Y}_2) \\
       R_{00}+R_{11}+B_1&\le I(W,U;Y_1)+I(W,U,Y_1;\hat{Y}_2)  \nonumber \\
	& \qquad \qquad +C_{21}-I(Y_2;\hat{Y}_2)\\
       R_{11}+B_1&\le I(U;Y_1,\hat{Y}_2|W)\\
       R_{00}+R_{11}+B_1&\le I(W,U;Y_1,\hat{Y}_2).\label{eq:Y1BndFinlast}
\end{align}
Moreover, the Markov chain condition $(U,W,Y_1) \rightarrow Y_2 \rightarrow \hat{Y}_2$ on the joint distribution implies that
%Moreover, note that the given probability distribution implies that 
\begin{multline}
        I(W,U,Y_1;\hat{Y}_2)+C_{21}-I(Y_2;\hat{Y}_2)= \\ C_{21}-I(\hat{Y}_2;Y_2|U,W,Y_1).
\end{multline}
One can easily see that if $C_{21}<I(\hat{Y}_2;Y_2|U,W,Y_1)$, then the optimal $\hat{Y}_2 = \varnothing$.
Otherwise, we keep $C_{21}-I(\hat{Y}_2;Y_2|U,W,Y_1)$ as is.

Then, (\ref{eq:Y1BndFin})-(\ref{eq:Y1BndFinlast}) are equivalent to
\begin{align}\label{eq:Y1last}
R_{11}&+B_1 \le \min \Big\{I(U;Y_1|W)  \nonumber \\ 
+&\{C_{21}-I(\hat{Y}_2;Y_2|U,W,Y_1)\}^+,I(U;Y_1,\hat{Y}_2|W)\Big\}\\
R_{00}&+R_{11} +B_1 \le \min \Big\{I(W,U;Y_1) \nonumber \\ 
+&\{C_{21}-I(\hat{Y}_2;Y_2|U,W,Y_1)\}^+,I(W,U;Y_1,\hat{Y}_2)\Big\}.\label{eq:Y1last_2}
\end{align}

Similarly, one can see that (\ref{eq:Cind})-(\ref{eq:ApBC12qhatmin}), (\ref{eq:R1hat}), (\ref{eq:Y2bnds1})-(\ref{eq:Y2bnds1_2}), and (\ref{eq:Y2bnds2})-(\ref{eq:Y2bnds2_2}) can be equivalently rewritten as 
\begin{align}\label{eq:Y2bnds}
R_{22}+&B_2 \le I(V;Y_2|W)+I(W,V,Y_2;\hat{Y}_1) \nonumber \\ 
	&+{\alpha}_1C_{12}-I(W,U,Y_1;\hat{Y}_1)\\
R_{00}+R_{22}+&B_2 \le I(W,V;Y_2)+I(W,V,Y_2;\hat{Y}_1) \nonumber \\ 
	&+{\alpha}_1C_{12}-I(W,U,Y_1;\hat{Y}_1)+\Bar{\alpha}_1{C}_{12}\\
R_{22}+&B_2 \le I(V;\hat{Y}_1,Y_2|W)\\
R_{00}+R_{22}+&B_2 \le I(W,V;\hat{Y}_1,Y_2)+\Bar{\alpha}_1{C}_{12}.\label{eq:Y2bnds_last}
\end{align}
In addition, the Markov chain condition $(V,Y_2) \rightarrow (Y_1,U,W) \rightarrow \hat{Y}_1$ on the joint distribution implies that
\begin{multline}
        I(W,V,Y_2;\hat{Y}_1)+{\alpha}_1C_{12}-I(W,U,Y_1;\hat{Y}_1)= \\
        \alpha_1 C_{12}-I(\hat{Y}_1;U,Y_1|V,W,Y_2).
\end{multline}
If $\alpha_1 C_{12}<I(\hat{Y}_1;U,Y_1|V,W,Y_2)$, then one can easily show that the optimal $\hat{Y}_1 = \varnothing$. Therefore, (\ref{eq:Y2bnds})-(\ref{eq:Y2bnds_last}) can be given in the following equivalent form:
\begin{align}\label{eq:Y2last}
R_{22}&+B_2 \le \min \Big\{I(V;Y_2|W)+\{\alpha_1 C_{12}- \nonumber \\
& I(\hat{Y}_1;U,Y_1|V,W,Y_2)\}^+, 
I(V;\hat{Y}_1,Y_2|W)\Big\} \\
R_{00}+R_{22}&+B_2 \le \min \Big\{I(W,V;Y_2) \nonumber \\
&+\{\alpha_1 C_{12}- I(\hat{Y}_1;U,Y_1|V,W,Y_2)\}^++\Bar{\alpha}_1{C}_{12}, \nonumber \\
& \ I(W,V;\hat{Y}_1,Y_2)+\Bar{\alpha}_1{C}_{12}\Big\}. \label{eq:Y2last_2}
\end{align}

Lastly, considering the bounds in (\ref{eq:ApBbiningbound}), (\ref{eq:Y1last})-(\ref{eq:Y1last_2}), and (\ref{eq:Y2last})-(\ref{eq:Y2last_2}) together, by applying the Fourier–Motzkin algorithm to eliminate all the partial rates, one can derive the rate region $\mathcal{R}_i^{(1)}$. The proof is thus complete.
%\blacksquare

We note that the error probability analysis here can also be derived from the general approach of \cite{chung}, but the general expressions in \cite{chung} are not always easy to simplify.

\section{Proof of Proposition \ref{th.ibc2}}

\label{app:achievability_2}

Fix a joint distribution $$P(u,v,w,x)P(\hat{y}_1|u,w,y_1)P(\hat{y}_2|w,y_2)$$ and let $\alpha_2 \in [0,1]$. Similar to the proof of Proposition \ref{th.ibc1}, we construct random codebooks of length-$n$ codewords to transmit the common message $M_0$, and the private messages $M_1$ and $M_2$, which are independent and uniformly distributed over the sets $[1:2^{nR_0}]$, $[1:2^{nR_1}]$, and $[1:2^{nR_2}]$, respectively. 

Each of the private messages $M_1$ and $M_2$ is split into two parts as follows:
\begin{align}
    M_1&\longmapsto (M_{10},M_{11}),
\end{align}
where $M_{10}\in [1:2^{nR_{10}}]$, $M_{11}\in [1:2^{nR_{11}}]$, and
\begin{align}
    M_2&\longmapsto (M_{20},M_{22}),
\end{align}
where $M_{20}\in [1:2^{nR_{20}}]$, $M_{22}\in [1:2^{nR_{22}}]$.
Therefore, 
\begin{align}
        R_1=R_{10}+R_{11}, \\
        R_2=R_{20}+R_{22}.
\end{align}
The sub-messages $M_{10}$ and $M_{20}$ are decoded at both receivers, i.e., the triple $(M_0,M_{10},M_{20})$ is considered as the common message with rate 
\begin{align}
        R_{00}=R_0+R_{10}+R_{20}.
\end{align}
Moreover, let $B_1, B_2, \hat{R}_1, \hat{R}_2$ be non-negative real numbers denoting 
the rates of the binning steps and the compression steps in the proof. Also, define
\begin{align}
        \hat{C}_{12}&=\min \{C_{12},\hat{R}_1\} \label{eq:ApCC12hat}\\
        \hat{C}_{21}^d&=\min \{\Bar{\alpha}_2C_{21},R_{00}\} \\
        \hat{C}_{21}^q&=\min \{{\alpha}_2C_{21},\hat{R}_2\} \label{eq:ApCC21qhat}
\end{align}
as the portion of the conferencing link rates used in the two directions and for decode-and-forward and quantize-bin-and-forward, respectively.

\subsection{Encoding at the Transmitter}

\begin{enumerate}
\item Generate at random $2^{nR_{00}}$ independent codewords $W^n$ according to $P(w^n)=\prod_{t=1}^n P(w_t)$. Label these codewords $W^n(m_{00})$ where $m_{00}\in [1:2^{nR_{00}}]$.

\item For each codeword $W^n(m_{00})$, generate at random $2^{n(R_{11}+B_1)}$ independent codewords $U^n$ according to $\prod_{t=1}^nP(u_t|w_t(m_{00}))$. Label these codewords $U^n(m_{00},m_{11},b_1)$ where $m_{11}\in[1:2^{nR_{11}}]$ and $b_1\in[1:2^{nB_1}]$.

\item For each codeword $W^n(m_{00})$, generate at random $2^{n(R_{22}+B_2)}$ independent codewords $V^n$ according to $\prod_{t=1}^nP(v_t|w_t(m_{00}))$. Label these codewords $V^n(m_{00},m_{22},b_2)$ where $m_{22}\in[1:2^{nR_{22}}]$ and $b_2\in[1:2^{nB_2}]$.

\item Given a triple $(m_{00},m_{11},m_{22})$ to be transmitted, let $(b_1^\mathcal{T},b_2^\mathcal{T})$ be a pair that satisfies 
\begin{multline}\label{Etypd}
    (W^n(m_{00}),U^n(m_{00},m_{11},b_1^\mathcal{T}),V^n(m_{00},m_{22},b_2^\mathcal{T}))\\ \in\mathcal{T}_\epsilon^n(P_{WUV}).
\end{multline}
If there are multiple pairs of $(b_1,b_2)$ that satisfy (\ref{Etypd}), 
let $(b_1^\mathcal{T},b_2^\mathcal{T})$ be the one with minimum $b_1$. 
If multiple such pairs have the same minimum $b_1$, choose the one with the minimum $b_2$. 
If no pair $(b_1,b_2)$ satisfies (\ref{Etypd}), declare an encoding error. 
		%let $(b_1^\mathcal{T},b_2^\mathcal{T})$ be the one with minimum $b_1+b_2$. If there is no such a pair, declare encoding error. 
Based on the mutual covering lemma \cite{NIT}, one can easily see that for sufficiently large $n$ and small $\epsilon$, the probability of encoding error tends to zero provided that
\begin{align}
         B_1+B_2\ge I(U;V|W).
\end{align}
The transmitter then generates a codeword $X^n$ according to $$\prod_{t=1}^nP(x_t|w_t(m_{00}),u_t(m_{00},m_{11},b_1^\mathcal{T}),v_t(m_{00},m_{22},b_2^\mathcal{T}))$$ and sends it over the channel.
\end{enumerate}

\subsection{Encoding and Decoding at the Receivers}

\begin{enumerate}
\item Partition the set $[1:2^{nR_{00}}]$ into $2^{n\hat{C}_{21}^d}$ bins, each containing $2^{n(R_{00}-\hat{C}_{21}^d)}$ elements. For a given $m_{00}\in [1:2^{nR_{00}}]$, let $\mathfrak{C}_{21}^d(m_{00})$ denote the index of the bin that $m_{00}$ belongs to.

\item Generate at random $2^{n\hat{R}_2}$ independent codewords $\hat{Y}_2^n$ according to the distribution $\prod_{t=1}^nP(\hat{y}_{2,t})$. Label these codewords $\hat{Y}_2^n(\hat{m}_2)$ where $\hat{m}_2\in [1:2^{n\hat{R}_2}]$. 

\item Partition the set $[1:2^{n\hat{R}_2}]$ into $2^{n\hat{C}_{21}^q}$ bins, each containing $2^{n(\hat{R}_2-\hat{C}_{21}^q)}$ elements. For a given $\hat{m}_2\in [1:2^{n\hat{R}_2}]$, let $\mathfrak{C}_{21}^q(\hat{m}_2)$ denote the index of the bin that $\hat{m}_2$ belongs to. 

\item Given its output sequence $Y_2^n$, the second receiver seeks a unique $m_{00}^\bullet$ for which there exists some $(m_{22},b_2)$ that satisfies 
\begin{multline}
        (W^n(m_{00}^\bullet),V^n(m_{00}^\bullet,m_{22},b_2),Y_2^n) \in \mathcal{T}_\epsilon^n(P_{WVY_2}).
\end{multline}
If there is no such $m_{00}^\bullet$ or if there is more than one, declare an error. The probability of this error event tends to zero provided that
\begin{align}
        R_{00}+R_{22}+B_2\le I(W,V;Y_2)
\end{align}

\item Let $\hat{m}_2^\mathcal{T}$ be the smallest integer in $[1:2^{n\hat{R}_2}]$ that satisfies 
\begin{multline}
        (W^n(m_{00}^\bullet),Y_2^n,\hat{Y}_2^n(\hat{m}_2^\mathcal{T}))\in \mathcal{T}_\epsilon^n(P_{WY_2\hat{Y}_2}). \label{eq:ApCtypWY2Y2hat}
\end{multline}
If there is no integer in $[1:2^{n\hat{R}_2}]$ that satisfies (\ref{eq:ApCtypWY2Y2hat}), then declare an error. The probability of this error tends to zero provided that
\begin{align}
        \hat{R}_2\ge I(W,Y_2;\hat{Y}_2)\label{eq:ApCIWY2Y2hat}
\end{align}
The receiver $Y_2$ then sends the pair $(\theta_1,\theta_2) = (\mathfrak{C}_{21}^d(m_{00}^\bullet),\mathfrak{C}_{21}^q(\hat{m}_2^\mathcal{T}))$ to the receiver $Y_1$ through the digital link.

\item Given $(\theta_1,\theta_2) \in [1:2^{n\hat{C}_{21}^d}]\times[1:2^{n\hat{C}_{21}^q}]$ and its output sequence $Y_1^n$, the first receiver seeks a unique triple $(m_{00}^*,m_{11}^*,b_1^*)$ for which there exists some $\hat{Y}_2^n$ that satisfies 
\begin{multline}
    (W^n(m_{00}^*),U^n(m_{00}^*,m_{11}^*,b_1^*),Y_1^n,\hat{Y}_2^n) \\ \in \mathcal{T}_\epsilon^n(P_{WUY_1\hat{Y}_2})\label{eq:ApCtypWUY1Y2hat}
\end{multline}
and
\begin{align}
  \mathfrak{C}_{21}^d(m_{00}^*)&=\theta_1, \label{eq:ApCm00st1}\\
  \mathfrak{C}_{21}^q(\hat{Y}_2^n)&=\theta_2.\label{eq:ApCY2hatt2}
\end{align}
If there is such a unique triple, the first receiver decodes its intended messages as $m_{00}^*$ and $m_{11}^*$. Otherwise, it declares an error. The probability that the actual transmitted messages do not satisfy (\ref{eq:ApCtypWUY1Y2hat})-(\ref{eq:ApCY2hatt2}) is negligible for sufficiently large $n$. Consider the error event that some messages other than the actual transmitted ones together with some $\hat{Y}_2^n$ would satisfy (\ref{eq:ApCtypWUY1Y2hat})-(\ref{eq:ApCY2hatt2}). This event comprises two cases. First, this $\hat{Y}_2^n$ is not the codeword actually selected from the quantization codebook, i.e., $\hat{Y}_2^n(\hat{m}_2^\mathcal{T})$. Second, this $\hat{Y}_2^n$ is actually $\hat{Y}_2^n(\hat{m}_2^\mathcal{T})$. Using the multivariate packing lemma \cite{NIT}, one can see that the probability of the first case tends to zero provided that
\begin{align}
\hat{R}_2-\hat{C}_{21}^q+&R_{11}+B_1\le \nonumber \\ 
	&I(U;Y_1|W)+I(W,U,Y_1;\hat{Y}_2)\label{eq:ApCIUY1WWUY1Y2hat}\\
\hat{R}_2-\hat{C}_{21}^q+&R_{00}-\hat{C}_{21}^d+R_{11}+B_1\le \nonumber \\ 
	&I(W,U;Y_1)+I(W,U,Y_1;\hat{Y}_2)
\end{align}
and the probability of the second case tends to zero provided that
\begin{align}
        R_{11}+B_1\le I(U;Y_1,\hat{Y}_2|W)\\
        R_{00}-\hat{C}_{21}^d+R_{11}+B_1\le I(W,U;Y_1,\hat{Y}_2).\label{eq:ApCIWUY1Y2hat}
\end{align}

\item Generate at random $2^{n\hat{R}_1}$ independent codewords $\hat{Y}_1^n$ according to the distribution $\prod_{t=1}^nP(\hat{y}_{1,t})$. Label these codewords $\hat{Y}_1^n(\hat{m}_1)$ where $\hat{m}_1\in [1:2^{n\hat{R}_1}]$. 

\item Partition the set $[1:2^{n\hat{R}_1}]$ into $2^{n\hat{C}_{12}}$ bins, each containing $2^{n(\hat{R}_1-\hat{C}_{12})}$ elements. For a given $\hat{m}_1\in [1:2^{n\hat{R}_1}]$, let $\mathfrak{C}_{12}(\hat{m}_1)$ denote the index of the bin that $\hat{m}_1$ belongs to. 

\item Given the output sequence $Y_1^n$ at the first receiver, let $\hat{m}_1^\mathcal{T}$ be the smallest integer in $[1:2^{n\hat{R}_1}]$ that satisfies 
\begin{multline}
        (W^n(m_{00}^*),U^n(m_{00}^*,m_{11}^*,b_1^*),Y_1^n,\hat{Y}_1^n(\hat{m}_1^\mathcal{T})) \\ \in \mathcal{T}_\epsilon^n(P_{WUY_1\hat{Y}_1}), \label{eq:ApCtypWUY1Y1hat}
\end{multline}
where $(m_{00}^*,m_{11}^*,b_1^*)$ is the unique triple decoded in Step 6 according to (\ref{eq:ApCtypWUY1Y2hat})-(\ref{eq:ApCY2hatt2}). If there is no integer $[1:2^{n\hat{R}_1}]$ that satisfies (\ref{eq:ApCtypWUY1Y1hat}), then declare an error. The probability of this error event tends to zero provided that
\begin{align}
        \hat{R}_1\ge I(W,U,Y_1;\hat{Y}_1)\label{eq:ApCR1hatge}
\end{align}
The receiver $Y_1$ then sends $\kappa= \mathfrak{C}_{12}(\hat{m}_1^\mathcal{T})$ to the receiver $Y_2$ through the digital link.

\item Given $\kappa\in [1:2^{n\hat{C}_{12}}]$ and its output sequence $Y_2^n$, the second receiver seeks a unique pair $(m_{22}^\bullet,b_2^\bullet)$ for which there exists some $\hat{Y}_1^n$ that satisfies 
\begin{multline}
    (W^n(m_{00}^\bullet),V^n(m_{00}^\bullet,m_{22}^\bullet,b_2^\bullet),\hat{Y}_1^n,Y_2^n) \\ \in \mathcal{T}_\epsilon^n(P_{WV\hat{Y}_1Y_2})\label{eq:ApCtypWVY1hatY2}
\end{multline}
and
\begin{align}
        \mathfrak{C}_{12}(\hat{Y}_1^n)=\kappa. \label{eq:ApCY1hatk}
\end{align}
Note that $m_{00}^\bullet$ has been already decoded in Step 4. If there is such a unique pair, the receiver decodes its intended (private) message as $m_{22}^\bullet$. Otherwise, it declares an error. The probability that the actual transmitted messages do not satisfy (\ref{eq:ApCtypWVY1hatY2})-(\ref{eq:ApCY1hatk}) is negligible for sufficiently large $n$. Consider the error event that some messages other than the actual transmitted ones together with some $\hat{Y}_1^n$ would satisfy (\ref{eq:ApCtypWVY1hatY2})-(\ref{eq:ApCY1hatk}). This event comprises two cases. First, this $\hat{Y}_1^n$ is not the codeword actually selected from the quantization codebook, i.e., $\hat{Y}_1^n(\hat{m}_1^\mathcal{T})$. Second, this $\hat{Y}_1^n$ is actually $\hat{Y}_1^n(\hat{m}_1^\mathcal{T})$. Using the multivariate packing lemma \cite{NIT}, one can show that the probability of the first case tends to zero provided that
\begin{align}
        \hat{R}_1-\hat{C}_{12}+R_{22}+B_2\le I(V;Y_2|W)+I(W,V,Y_2;\hat{Y}_1)\label{eq:ApCIVY2WWVY2Y1hat}
\end{align}
and the probability of the second case tends to zero provided that
\begin{align}
        R_{22}+B_2\le I(V;\hat{Y}_1,Y_2|W).\label{eq:ApCIVY1hatY2W}
\end{align}
\end{enumerate}

\subsection{Final Analysis}

By using Fourier-Motzkin elimination and by algebraic computation, one can see that (\ref{eq:ApCC12hat})-(\ref{eq:ApCC21qhat}), (\ref{eq:ApCIWY2Y2hat}), and (\ref{eq:ApCIUY1WWUY1Y2hat})-(\ref{eq:ApCIWUY1Y2hat}) can be equivalently rewritten in the following form: 
\begin{align}
        R_{11}+B_1&\le I(U;Y_1|W)+I(W,U,Y_1;\hat{Y}_2) \nonumber \\
& \ \ +{\alpha}_2C_{21}-I(W,Y_2;\hat{Y}_2)\label{eq:ApCR11B1}\\
        R_{00}+R_{11}+B_1&\le I(W,U;Y_1)+I(W,U,Y_1;\hat{Y}_2) \nonumber \\
& \ \ +{\alpha}_2C_{21}-I(W,Y_2;\hat{Y}_2)+\Bar{\alpha}_2C_{21}\\
        R_{11}+B_1&\le I(U;Y_1,\hat{Y}_2|W)\\
        R_{00}+R_{11}+B_1&\le I(W,U;Y_1,\hat{Y}_2)+\Bar{\alpha}_2C_{21}. \label{eq:ApCR00R11B1}
\end{align}
Moreover, the Markov chain condition $(U,Y_1) \rightarrow (W,Y_2) \rightarrow \hat{Y}_2$ on the joint distribution implies that
\begin{multline}
        I(W,U,Y_1;\hat{Y}_2)+{\alpha}_2C_{21}-I(W,Y_2;\hat{Y}_2)= \\ {\alpha}_2C_{21}-I(\hat{Y}_2;Y_2|U,W,Y_1).
\end{multline}
One can easily see that if ${\alpha}_2C_{21}<I(\hat{Y}_2;Y_2|U,W,Y_1)$, then the optimal $\hat{Y}_2 = \varnothing$. 

Thus, (\ref{eq:ApCR11B1})-(\ref{eq:ApCR00R11B1}) are equivalent to 
\begin{align}
R_{11}+B_1 &\le \min \Big\{I(U;Y_1|W) \nonumber \\
& \qquad + \{{\alpha}_2C_{21} -I(\hat{Y}_2;Y_2|U,W,Y_1)\}^+, \nonumber \\
& \qquad \qquad I(U;Y_1,\hat{Y}_2|W)\Big\}\label{eq:ApCR11B1min}\\
R_{00}+R_{11}+B_1 &\le \min \Big\{I(W,U;Y_1)+\{{\alpha}_2C_{21} \nonumber \\ 
& \qquad -I(\hat{Y}_2;Y_2|U,W,Y_1)\}^++\Bar{\alpha}_2C_{21}, \nonumber \\
& \qquad \qquad I(W,U;Y_1,\hat{Y}_2)+\Bar{\alpha}_2C_{21}\Big\}. \label{eq:ApCR00R11B1min}
\end{align}
Similarly, by algebraic computation, one can see that (\ref{eq:ApCC12hat})-(\ref{eq:ApCC21qhat}), (\ref{eq:ApCR1hatge}), (\ref{eq:ApCIVY2WWVY2Y1hat}), and (\ref{eq:ApCIVY1hatY2W}) can be rewritten in  the following equivalent form:
\begin{align}
R_{22}+B_2&\le I(V;Y_2|W)+I(W,V,Y_2;\hat{Y}_1)+C_{12} \nonumber \\
	& \qquad \qquad -I(W,U,Y_1;\hat{Y}_1)\label{eq:ApCR22B2C12-}\\
R_{22}+B_2&\le I(V;\hat{Y}_1,Y_2|W)\\
R_{00}+R_{22}+B_2&\le I(W,V;Y_2). \label{eq:ApCR00R22B2IWVY2}
\end{align}
The Markov chain condition $(V,Y_2) \rightarrow (U,W,Y_1) \rightarrow \hat{Y}_1$ on the joint distribution implies that
\begin{multline}
        I(W,V,Y_2;\hat{Y}_1)+C_{12}-I(W,U,Y_1;\hat{Y}_1)=\\ C_{12}-I(\hat{Y}_1;U,Y_1|V,W,Y_2).
\end{multline}
If $C_{12}<I(\hat{Y}_1;U,Y_1|V,W,Y_2)$, then the optimal $\hat{Y}_1 = \varnothing$. Therefore, (\ref{eq:ApCR22B2C12-})-(\ref{eq:ApCR00R22B2IWVY2}) are equivalent to 
\begin{align}
R_{22} +B_2 &\le \min \Big \{I(V;Y_2|W)\nonumber \\ 
&\qquad +\{C_{12}-I(\hat{Y}_1;U,Y_1|V,W,Y_2)\}^+, \nonumber \\
&\quad \qquad I(V;\hat{Y}_1,Y_2|W)\Big\}\label{eq:ApCR22B2min}\\
R_{00} +R_{22}+B_2 &\le I(W,V;Y_2). \label{eq:ApCR00R22B2le2}
\end{align}
Finally, considering the bounds (\ref{eq:ApCR11B1min})-(\ref{eq:ApCR00R11B1min}) and (\ref{eq:ApCR22B2min})-(\ref{eq:ApCR00R22B2le2}) together, by applying the Fourier–Motzkin algorithm to eliminate all the partial rates, we obtain the rate region $\mathcal{R}_i^{(2)}$. The proof is thus complete. %\blacksquare

\section{Proof of Theorem \ref{th.ogbc}}
\label{app:outer_bound_gbc}

We present the optimization over the auxiliary random variable $V$ in $\mathcal{R}_o$ in Theorem \ref{th.obc}. The constraints involving the auxiliary random variable $U$ can be evaluated similarly.
First note that since $|a|\ge|b|$, the receiver $Y_1$ is less noisy than the receiver $Y_2$, i.e.,
\begin{align}
	I(V;Y_2)\le I(V;Y_1) \qquad \forall P(u,x).
\end{align}
Therefore, we have
\begin{align}
%\begin{split}
    I(&X;Y_1|Y_2,V)+I(X;Y_2) \nonumber \\
    &=I(X;Y_1|Y_2,V)+I(X;Y_2|V)+I(V;Y_2)\\
    &=I(X;Y_1,Y_2|V)+I(V;Y_2)\\
    &\le I(X;Y_1,Y_2|V)+I(V;Y_1)\\
    &=I(X;Y_2|Y_1,V)+I(X;Y_1|V)+I(V;Y_1)\\
    &=I(X;Y_2|Y_1,V)+I(X;Y_1).
%\end{split}
\end{align}
In other words, constraint (\ref{eq:outer_R1V_2}) is inactive. 

Considering (\ref{eq:outer_R1V_1}), (\ref{eq:outer_R2V}), (\ref{eq:outer_RsumV_1}), (\ref{eq:outer_RsumV_2}), and (\ref{eq:outer_cs}) and accounting for the Markov chain $V \rightarrow X \rightarrow (Y_1,Y_2)$, it is easy to see that we need to maximize the following terms simultaneously:
\begin{align}
    I_1&\doteq I(X;Y_1,Y_2|V)+I(V;Y_2) \label{eq:ApI_1}\\
    I_2&\doteq I(V;Y_2) \label{eq:ApI_2}\\
    I_3&\doteq I(X;Y_1|V)+I(V;Y_2) \label{eq:ApI_3}\\
    I_4&\doteq I(X;Y_1,Y_2). \label{eq:ApI_4}
\end{align}

Consider first the special case where $\lambda = \frac{b}{a}$. 
In this case, $Z_1$ and $Z_2$ can be thought of as being related as 
\begin{equation}
Z_2 = \frac{b}{a} Z_1 + \tilde{Z},
\end{equation} 
where $Z_1$, $\tilde{Z}$ are independent. 
Since $Y_1 = a X + Z_1$, we have 
\begin{align}
Y_2 & = b X + Z_2 \\
& = b X + \frac{b}{a} Z_1 + \tilde{Z} \\
%& = \frac{b}{a} (a X + Z_1) + \tilde{Z} \\
& = \frac{b}{a} Y_1 + \tilde{Z}.
\end{align}
So, the BC is physically degraded (because $|a| \ge |b|$), i.e., $X \rightarrow Y_1 \rightarrow Y_2$ forms a Markov chain. 

In this special case, the simultaneous maximization of the four mutual information terms 
(\ref{eq:ApI_1})-(\ref{eq:ApI_4}) then reduces to the simultaneous maximization of $h(Y_1|V)$ and $h(Y_1)$ together with
minimization of $h(Y_2|V)$, which is achieved by jointly Gaussian $(X,V)$ as can be shown by the entropy power inequality in the proof of the capacity region of the degraded Gaussian broadcast channel \cite{NIT}. 

In the rest of the proof, we assume $\lambda \neq \frac{b}{a}$. 
Define a new virtual output $\hat{Y}$ (which is a scaled version
of the minimum mean-squared error (MMSE) estimate of $X$ given $(Y_1,Y_2)$) as follows:
\begin{align}
    \hat{Y}&\doteq (a-\lambda b)Y_1+(b-\lambda a)Y_2\\
    &=\left((a-\lambda b)a+(b-\lambda a)b\right)X+\hat{Z}\\
    &=(a^2+b^2-2\lambda ab)X+\hat{Z}
\end{align}
where 
\begin{align}
    \hat{Z}\doteq(a-\lambda b)Z_1+(b-\lambda a)Z_2.
\end{align}
When $\lambda \neq \frac{b}{a}$, it is clear
that there is a one-to-one mapping between $(\hat{Y},Y_2)$ and $(Y_1,Y_2)$. 
Moreover, $\lambda \neq \frac{b}{a}$ and the fact
$|\lambda|\le 1$ imply that 
$a^2+b^2-2\lambda ab \neq 0$, so we can write
\begin{align}
    Y_2=\frac{b}{a^2+b^2-2\lambda ab}\hat{Y}+\hat{Z_2},
\end{align}
where
\begin{align}
    \hat{Z_2}\doteq \frac{-b\hat{Z}}{a^2+b^2-2\lambda ab}+Z_2.
\end{align}

It can be verified that $\hat{Z_2}$ and $\hat{Z}$ are uncorrelated, and therefore they are independent (as they are jointly Gaussian). This means that the output $Y_2$ is in fact a noisy %(i.e., physically degraded) 
version of $\hat{Y}$, i.e., $X \rightarrow \hat{Y} \rightarrow Y_2$ forms a Markov chain. This relationship can also be recognized from the fact that $\hat{Y}$ is a (scaled) MMSE estimate of $X$ given $(Y_1,Y_2)$.  Thus, we have
\begin{align} 
I(X;Y_1,Y_2|V)&=I(X;\hat{Y},Y_2|V)=I(X;\hat{Y}|V), \label{eq:ApYhatV}\\
I(X;Y_1,Y_2)&=I(X;\hat{Y},Y_2)=I(X;\hat{Y}). \label{eq:ApY1Y2}
\end{align}

%The maximum value of the mutual information term $I(X;\hat{Y})$ can be easily obtained based on the fact that Gaussian random variable $X \sim \mathcal{N}(0,P)$ maximizes the entropy. 

Now, utilizing (\ref{eq:ApYhatV})-(\ref{eq:ApY1Y2}) for bounding the maximum
values of (\ref{eq:ApI_1})-(\ref{eq:ApI_4}), we see that after 
expressing the mutual information in terms of differential entropy, the optimization 
is effectively to maximize $h(\hat{Y}|V)$, $h(Y_1|V)$ and $h(\hat{Y})$, and simultaneously to minimize $h(Y_2|V)$. %, where $h(\cdot)$ denotes the differential entropy. 

We complete this step using the entropy power inequality. First,
\begin{align}
    \frac{1}{2} \log (2\pi e)&=h(Z_1) \\
    &=h(Y_1|X,V)\\
    &\le h(Y_1|V) \label{eq:ApY1V_left} \\
    &\le h(Y_1) \\ %=aX+Z_1)\\ 
&\le \frac{1}{2} \log 2\pi e(a^2P+1)\label{eq:ApY1V}.
\end{align}
By comparing the two sides of (\ref{eq:ApY1V_left})-(\ref{eq:ApY1V}), we find that there is $\beta \in [0,1]$ such that
\begin{align}
    h(Y_1|V)=\frac{1}{2} \log 2\pi e(\beta a^2P+1). \label{eq:EPI_1}
\end{align}
Now, let $\hat{Z}^*$ be a Gaussian random variable independent of all other variables and with zero mean and a variance equal to $\frac{(a\lambda-b)^2}{\mu}$ where $\mu\doteq a^2+b^2-2\lambda ab$. We then have
\begin{align}
        h(\hat{Y}|V)&=h(\mu X+\hat{Z}) \\
        &=h\left(aX+\frac{a}{\mu}\hat{Z}\right)-\frac{1}{2} \log \left(\frac{a^2}{\mu^2}\right)\\
        &\le \frac{1}{2} \log \left(e^{2h(aX+\frac{a}{\mu}\hat{Z}+\hat{Z}^*|V)} - e^{2h(\hat{Z}^*|V)} \right) \nonumber \\
        & \qquad \qquad \qquad -\frac{1}{2} \log \left(\frac{a^2}{\mu^2}\right) \label{eq:converse_epi_0} \\
        &=\frac{1}{2} \log \left(e^{2h(aX+Z_1|V)} - e^{2h(\hat{Z}^*)}\right)  \nonumber \\
        & \qquad \qquad \qquad -\frac{1}{2} \log \left(\frac{a^2}{\mu^2}\right) \label{eq:converse_epi_0b}\\
        &=\frac{1}{2} \log \left(e^{2h(Y_1|V)} - e^{2h(\hat{Z}^*)} \right)  \nonumber \\
        & \qquad \qquad \qquad -\frac{1}{2} \log \left(\frac{a^2}{\mu^2}\right)\\
        &=\frac{1}{2} \log \left(2\pi e \left(\beta \mu P+1-\lambda^2\right)\mu\right)
\label{eq:EPI_2}
\end{align}
where (\ref{eq:converse_epi_0}) is due to the entropy power inequality, and (\ref{eq:converse_epi_0b}) holds because $\frac{a}{\mu}\hat{Z}+\hat{Z}^*$ is a Gaussian random variable with zero mean and unit variance and independent of $X$ and $V$, therefore it can be replaced by $Z_1$ in the differential entropy expression, and also $\hat{Z}^*$ is independent of $V$.

Finally, let $Z_2^*$ be a Gaussian random variable independent of all other variables and with zero mean and variance $\frac{a^2-b^2}{b^2}$, we have
\begin{align}
        h(Y_2|V)&=h(bX+Z_2|V) \nonumber \\
        &=h\left(aX+\frac{a}{b}Z_2|V\right)-\frac{1}{2} \log\left(\frac{a^2}{b^2}\right)\\
        &=h(aX+Z_2+Z_2^*|V)-\frac{1}{2} \log\left(\frac{a^2}{b^2}\right)\\
        &\ge \frac{1}{2} \log \left(e^{2h(aX+Z_2|V)}+e^{2h(Z_2^*|V)}\right) \nonumber \\
        & \qquad \qquad \qquad \qquad -\frac{1}{2} \log\left(\frac{a^2}{b^2}\right) \label{eq:converse_epi} \\
        &=\frac{1}{2} \log\left(e^{2h(aX+Z_1|V)}+e^{2h(Z_2^*)}\right) \nonumber \\
        & \qquad \qquad \qquad \qquad -\frac{1}{2} \log\left(\frac{a^2}{b^2}\right) \label{eq:converse_epi_b} \\
        &=\frac{1}{2} \log\left(e^{2h(Y_1|V)}+e^{2h(Z_2^*)}\right) \nonumber \\
        & \qquad \qquad \qquad \qquad -\frac{1}{2} \log\left(\frac{a^2}{b^2}\right)\\
        &=\frac{1}{2}\log\left(2\pi e\left(\beta b^2P+1\right)\right)
\label{eq:EPI_3}
\end{align}
where (\ref{eq:converse_epi}) is due to the entropy power inequality, and (\ref{eq:converse_epi_b}) holds because $Z_2$ has the same distribution as $Z_1$ and both are independent of $(X,V)$, and also $Z_2^*$ is independent of $V$.

Considering \eqref{eq:EPI_1}, \eqref{eq:EPI_2}, and \eqref{eq:EPI_3}
together with the fact that a Gaussian $X \sim \mathcal{N}(0,P)$ maximizes $h(\hat{Y})$
allows us to conclude that a jointly Gaussian $(X,V)$ simultaneously maximizes the mutual information terms in (\ref{eq:ApI_1})-(\ref{eq:ApI_4}) and (\ref{eq:ApYhatV})-(\ref{eq:ApY1Y2}). 

The proof can be completed by substituting the evaluation of these mutual information terms (\ref{eq:ApI_1})-(\ref{eq:ApI_4}) for Gaussian $(X,V)$ into $\mathcal{R}_o$ in Theorem \ref{th.obc}.
%The proof is thus complete. %\blacksquare
%The maximum value of the mutual information term $I(X;\hat{Y})$ can be easily obtained based on the fact that Gaussian random variable $X \sim \mathcal{N}(0,P)$ maximizes the entropy. 

%\section{Proof of Theorem \ref{th.FCgbc2}}
%
%\label{app:theorem_8}

\bibliographystyle{IEEEtran}
\bibliography{IEEEabrv,references}

% Generated by IEEEtran.bst, version: 1.14 (2015/08/26)
\begin{thebibliography}{10}
\providecommand{\url}[1]{#1}
\csname url@samestyle\endcsname
\providecommand{\newblock}{\relax}
\providecommand{\bibinfo}[2]{#2}
\providecommand{\BIBentrySTDinterwordspacing}{\spaceskip=0pt\relax}
\providecommand{\BIBentryALTinterwordstretchfactor}{4}
\providecommand{\BIBentryALTinterwordspacing}{\spaceskip=\fontdimen2\font plus
\BIBentryALTinterwordstretchfactor\fontdimen3\font minus
  \fontdimen4\font\relax}
\providecommand{\BIBforeignlanguage}[2]{{%
\expandafter\ifx\csname l@#1\endcsname\relax
\typeout{** WARNING: IEEEtran.bst: No hyphenation pattern has been}%
\typeout{** loaded for the language `#1'. Using the pattern for}%
\typeout{** the default language instead.}%
\else
\language=\csname l@#1\endcsname
\fi
#2}}
\providecommand{\BIBdecl}{\relax}
\BIBdecl

\bibitem{reza_itw}
R.~K. Farsani and W.~Yu, ``Capacity bounds for broadcast channels with
  bidirectional conferencing decoders,'' in \emph{IEEE Inf. Theory Workshop
  (ITW)}, Apr. 2023, pp. 59--63.

\bibitem{reza_isit}
------, ``Gaussian broadcast channels with bidirectional conferencing decoders
  and correlated noises,'' in \emph{IEEE Inter. Symp. Inf. Theory (ISIT)}, Jul.
  2023, pp. 1514--1519.

\bibitem{Wimac}
F.~Willems, ``The discrete memoryless multiple access channel with partially
  cooperating encoders,'' \emph{{IEEE} Trans. Inf. Theory}, vol.~29, no.~3, pp.
  441--445, May 1983.

\bibitem{Madsen}
A.~H\o{}st-Madsen, ``Capacity bounds for cooperative diversity,'' \emph{{IEEE}
  Trans. Inf. Theory}, vol.~52, no.~4, pp. 1522--1544, Apr. 2006.

\bibitem{CaoChen}
Y.~Cao and B.~Chen, ``An achievable rate region for interference channels with
  conferencing,'' in \emph{IEEE Inter. Symp. Inf. Theory (ISIT)}, Jun. 2007,
  pp. 1251--1255.

\bibitem{Maric}
I.~Maric, R.~D. Yates, and G.~Kramer, ``Capacity of interference channels with
  partial transmitter cooperation,'' \emph{{IEEE} Trans. Inf. Theory}, vol.~53,
  no.~10, pp. 3536--3548, Oct. 2007.

\bibitem{Gunduz}
O.~Simeone, D.~Gündüz, H.~V. Poor, A.~J. Goldsmith, and S.~Shamai, ``Compound
  multiple-access channels with partial cooperation,'' \emph{{IEEE} Trans. Inf.
  Theory}, vol.~55, no.~6, pp. 2425--2441, Jun. 2009.

\bibitem{SkoglundA}
H.~T. Do, T.~J. Oechtering, and M.~Skoglund, ``An achievable rate region for
  the {Gaussian} {Z}-interference channel with conferencing,'' in
  \emph{Allerton Conf. Commun., Control, Computing}, Sept./Oct. 2009, pp.
  75--81.

\bibitem{SkoglundI}
------, ``The {Gaussian} {Z}-interference channel with rate-constrained
  conferencing decoders,'' in \emph{IEEE Inter. Conf. Commun.}, May 2010, pp.
  1--5.

\bibitem{Rezamac1}
A.~Haghi, R.~Khosravi-Farsani, M.~R. Aref, and F.~Marvasti, ``The capacity
  region of fading multiple access channels with cooperative encoders and
  partial csit,'' in \emph{IEEE Inter. Symp. Inf. Theory (ISIT)}, Jun. 2010,
  pp. 485--489.

\bibitem{Tse1}
I.-H. Wang and D.~N.~C. Tse, ``Interference mitigation through limited
  transmitter cooperation,'' \emph{{IEEE} Trans. Inf. Theory}, vol.~57, no.~5,
  pp. 2941--2965, May 2011.

\bibitem{Tse2}
------, ``Interference mitigation through limited receiver cooperation,''
  \emph{{IEEE} Trans. Inf. Theory}, vol.~57, no.~5, pp. 2913--2940, May 2011.

\bibitem{Kang}
N.~Liu, D.~Gündüz, and W.~Kang, ``Capacity results for a class of
  deterministic {Z}-interference channels with unidirectional receiver
  conferencing,'' in \emph{6th Inter. ICST Conf. Commun. Netw. China
  (CHINACOM)}, Aug. 2011, pp. 580--584.

\bibitem{Rezamac2}
A.~Haghi, R.~Khosravi-Farsani, M.~R. Aref, and F.~Marvasti, ``The capacity
  region of $p$-transmitter/$q$-receiver multiple-access channels with common
  information,'' \emph{{IEEE} Trans. Inf. Theory}, vol.~57, no.~11, pp.
  7359--7376, Nov. 2011.

\bibitem{Wei}
L.~Zhou and W.~Yu, ``Gaussian {Z}-interference channel with a relay link:
  Achievability region and asymptotic sum capacity,'' \emph{{IEEE} Trans. Inf.
  Theory}, vol.~58, no.~4, pp. 2413--2426, Apr. 2012.

\bibitem{SkoglundT}
H.~T. Do, T.~J. Oechtering, and M.~Skoglund, ``On asymmetric interference
  channels with cooperating receivers,'' \emph{{IEEE} Trans. Inf. Theory},
  vol.~61, no.~2, pp. 554--563, Feb. 2013.

\bibitem{RezaK1}
R.~K. Farsani and A.~K. Khandani, ``Novel outer bounds and capacity results for
  the interference channel with conferencing receivers,'' in \emph{IEEE Inter.
  Symp. Inf. Theory (ISIT)}, Jun. 2017, pp. 649--653.

\bibitem{RezaK2}
------, ``Novel outer bounds and capacity results for the interference channel
  with conferencing receivers,'' \emph{{IEEE} Trans. Inf. Theory}, vol.~66,
  no.~6, pp. 3327--3341, Jun. 2020.

\bibitem{Draper}
S.~C. Draper, B.~J. Frey, and F.~R. Kschischang, ``Interactive decoding of a
  broadcast message,'' in \emph{Allerton Conf. Commun., Control, Computing},
  Oct. 2003.

\bibitem{Dabora}
R.~Dabora and S.~D. Servetto, ``Broadcast channels with cooperating decoders,''
  \emph{{IEEE} Trans. Inf. Theory}, vol.~52, no.~12, pp. 5438--5454, Dec. 2006.

\bibitem{bross}
S.~I. Bross, ``On the discrete memoryless partially cooperative relay broadcast
  channel and the broadcast channel with cooperating decoders,'' \emph{{IEEE}
  Trans. Inf. Theory}, vol.~55, no.~5, pp. 2161--2182, 2009.

\bibitem{Dik}
L.~Dikstein, H.~H. Permuter, and Y.~Steinberg, ``On state-dependent degraded
  broadcast channels with cooperation,'' \emph{{IEEE} Trans. Inf. Theory},
  vol.~62, no.~5, pp. 2308--2323, May 2016.

\bibitem{Gold}
Z.~Goldfeld, H.~H. Permuter, and G.~Kramer, ``Duality of a source coding
  problem and the semi-deterministic broadcast channel with rate-limited
  cooperation,'' \emph{{IEEE} Trans. Inf. Theory}, vol.~62, no.~5, pp.
  2285--2307, May 2016.

\bibitem{Cuff}
Z.~Goldfeld, G.~Kramer, H.~H. Permuter, and P.~Cuff, ``Strong secrecy for
  cooperative broadcast channels,'' \emph{{IEEE} Trans. Inf. Theory}, vol.~63,
  no.~1, pp. 469--495, Jan. 2017.

\bibitem{absent}
W.~Huleihel and Y.~Steinberg, ``Channels with cooperation links that may be
  absent,'' \emph{{IEEE} Trans. Inf. Theory}, vol.~63, no.~9, pp. 5886--5906,
  Sep. 2017.

\bibitem{unreliable}
D.~Itzhak and Y.~Steinberg, ``The broadcast channel with degraded message sets
  and unreliable conference,'' \emph{{IEEE} Trans. Inf. Theory}, vol.~67,
  no.~9, pp. 5623--5650, Sep. 2021.

\bibitem{marton}
J.~Körner and K.~Marton, ``Images of a set via two channels and their role in
  multi-user communication,'' \emph{{IEEE} Trans. Inf. Theory}, vol.~23, no.~6,
  pp. 751--761, 1977.

\bibitem{CsisKo}
I.~Csiszár and J.~Körner, ``Broadcast channels with confidential messages,''
  \emph{{IEEE} Trans. Inf. Theory}, vol.~24, no.~3, pp. 339--348, May 1978.

\bibitem{NIT}
A.~El~Gamal and Y.-H. Kim, \emph{Network Information Theory}.\hskip 1em plus
  0.5em minus 0.4em\relax Cambridge, U.K.: Cambridge Univ. Press, 2011.

\bibitem{FarsaniThesis}
R.~K. Farsani, ``Capacity bounds for broadcast channels with cooperative
  users,'' Ph.D. dissertation, University of Toronto, Canada, 2025.

\bibitem{Kimpr}
Y.-H. Kim, ``Coding techniques for primitive relay channels,'' in
  \emph{Allerton Conf. Commun., Control, Computing}, Sep. 2007.

\bibitem{gamgoh}
A.~El~Gamal, A.~Gohari, and C.~Nair, ``Achievable rates for the relay channel
  with orthogonal receiver components,'' in \emph{IEEE Inf. Theory Workshop
  (ITW)}, Oct. 2021, pp. 1--6.

\bibitem{Kimd}
Y.-H. Kim, ``Capacity of a class of deterministic relay channels,''
  \emph{{IEEE} Trans. Inf. Theory}, vol.~54, no.~3, pp. 1328--1329, 2008.

\bibitem{chung}
S.-H. Lee and S.-Y. Chung, ``A unified random coding bound,'' \emph{{IEEE}
  Trans. Inf. Theory}, vol.~64, no.~10, pp. 6779--6802, 2018.

\end{thebibliography}

\end{document}